\documentclass[12pt]{article}

\usepackage{float}
\usepackage[makeroom]{cancel}
\usepackage{color}
\usepackage{graphicx}
\usepackage{amsmath}
\usepackage{amssymb}
\usepackage{xspace}
\usepackage[small]{subfigure}
\usepackage[numbers,compress]{natbib}
\usepackage[hyperfootnotes=false]{hyperref}
\usepackage{tcolorbox}

\usepackage{framed}
\usepackage{physics}
\usepackage{tensor}

\newlength{\fighskip} \fighskip=2pt
\newlength{\figvskip} \figvskip=3pt

\newcommand*{\figbox}[2]{{
  \def\figscale{#1}
  \def\arraystretch{0.8}
  \arraycolsep=0pt
  \begin{array}{c}
    \vbox{\vskip\figscale\figvskip
      \hbox{\hskip\figscale\fighskip
        \includegraphics[scale=\figscale]{#2}}}
  \end{array}}}

\usepackage{mciteplus} 
\usepackage{dcolumn}
\usepackage{bm}
\usepackage{verbatim}
\usepackage{amscd}
\usepackage{amsfonts}
\usepackage{setspace}
\usepackage{amsthm}
\usepackage{enumerate}
\usepackage{mathtools}

\theoremstyle{plain}
\newtheorem{theorem}{Theorem}
\newtheorem{lemma}[theorem]{Lemma}

\theoremstyle{definition}

\newcommand{\Area}{\mathrm{Area}}

\usepackage{authblk}

\title{
\vspace{-80pt}
\hfill
{\normalsize RUP-26-24}\\
\vspace{40pt}
\bf 
Relative entropy of entanglement \\
and 
tripartite minimal surface 
}
\author[1,2]{Takato Mori\thanks{takato.mori@yukawa.kyoto-u.ac.jp}}
\author[3,4]{Beni Yoshida\thanks{byoshida@perimeterinstitute.ca}}
\affil[1]{\em \small Department of Physics, Rikkyo University, \protect\\
3-34-1 Nishi-Ikebukuro, Toshima-ku, Tokyo 171-8501, Japan}
\affil[2]{\em \small Transformative Research Innovation Platform of RIKEN platforms (TRIP) Headquarters,\par RIKEN, Wako 351-0198, Japan}
\affil[3]{\em \small Perimeter Institute for Theoretical Physics, Waterloo, Ontario N2L 3W8, Canada}
\affil[4]{\em \small 
Fundamental Quantum Science Program (FQSP), RIKEN, Wako 351-0198, Japan}
\date{}

\begin{document}
\maketitle

\begin{abstract}
We propose a holographic dual of the relative entropy of entanglement for a boundary tripartition $A:B:C$. 
Our proposal is that, at the leading order,
\begin{align}
E_R(\rho_{AB})
=
\frac{1}{4G_N}\Big[\mathrm{Area}(\Gamma_{\min})
-
\mathrm{Area}(\gamma_{AB})\Big]
\notag
\end{align}
where $\Gamma_{\min}$ is the minimal bulk tripartition surface and $\gamma_{AB}$ is the minimal surface homologous to $AB$.
For a general random tensor network, we prove the corresponding upper bound by constructing a separable state associated with $\Gamma_{\min}$ and evaluating its relative entropy. 
We also establish the matching lower bound rigorously for networks of one, two, and three Haar random tensors by bounding the maximal tripartite product-state overlap. 
In appropriate geometries, the minimal tripartition surface can form a nontrivial tri-junction resembling the Mercedes logo. \\

\begin{center}
\includegraphics[width=0.39\textwidth]{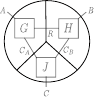}
\end{center}
\end{abstract}

\newpage

\tableofcontents

\newpage

\section{Introduction}

The Ryu--Takayanagi formula revealed a remarkable connection between quantum entanglement and geometry in holographic theories. 
Since then, a variety of quantum-information quantities have been considered, and candidate geometric bulk duals have been proposed for them~\cite{Nguyen2018, TakayanagiUmemoto2018, Cheng2020,
DuttaFaulkner2021, DongQiWalter2021, Mori2025QuantumCorrelation} (see~\cite{Horodecki:2009zz} for a review of entanglement measure).
Understanding such relations is important not only for clarifying the role of entanglement in the emergence of semiclassical spacetime, but also for addressing conceptual puzzles in quantum gravity in which quantum entanglement plays a fundamental role.
While establishing such duality relations directly in holographic theories is often technically challenging, holographic toy models such as Haar random tensor networks often provide a useful controlled setting~\cite{Pastawski:2015qua,Hayden:2016cfa}. 

Despite the ever-growing list of proposed geometric duals, concrete analytic understanding remains limited.
In fact, much of the rigorously established progress in holographic entanglement has so far focused on R\'{e}nyi and von Neumann entropies~\cite{LewkowyczMaldacena2013, Dong:2016fnf}, together with a relatively small class of related quantities such as logarithmic negativity~\cite{DongQiWalter2021}.
These quantities are comparatively tractable because they can be accessed through spectral properties and replica constructions.

By contrast, a number of fundamental entanglement measures still remain less understood. 
Examples include the entanglement of formation $E_F$~\cite{Hayden:2005sqo, Umemoto2019, MoriYoshidaConnectedWedge}, the relative entropy of entanglement $E_R$~\cite{MoriYoshidaREE}, the entanglement of purification $E_P$~\cite{Nguyen2018,TakayanagiUmemoto2018,Cheng2020}, and distillable entanglement $E_D$~\cite{MoriYoshidaConnectedWedge}. 
These quantities are defined through nontrivial variational or operational optimizations, which makes the usual replica techniques less effective.
Even for a single Haar-random state, some of these optimization-based measures were determined only relatively recently~\cite{MoriYoshidaREE,LiMoriYoshida2026}.

\paragraph{Relative entropy of entanglement:}

In this work, we study the relative entropy of entanglement (REE)~\cite{VedralPlenio1997}. 
For a bipartite state $\rho_{AB}$, it is defined by
\begin{align}
E_R(\rho_{A:B})
:=
\min_{\sigma_{AB}\in {\rm SEP}(A:B)}
D(\rho_{AB}\Vert \sigma_{AB}),
\label{eq:REE-def}
\end{align}
where
$D(\rho\Vert\sigma)
:=
\Tr\! \left[
\rho\bigl(\log \rho-\log \sigma\bigr)
\right]$
is the relative entropy.
Here, ${\rm SEP}(A:B)$ denotes the set of states separable across the bipartition $A:B$,
\begin{align}
\sigma_{AB} = 
\sum_i p_i
(\sigma_A^{(i)}\otimes\sigma_B^{(i)}), \qquad p_i >0,\quad \sum_i p_i=1.
\label{eq:SEP-def}
\end{align}
Thus, the REE quantifies how distinguishable a given state is from the entire set of separable states. 
It is also related to other important entanglement measures~\cite{Horodecki:2009zz}:
\begin{align}
E_D(\rho_{A:B})
\leq E_R(\rho_{A:B})
\leq E_F(\rho_{A:B}),
\qquad
E_R(\rho_{A:B})
\leq I(A:B)
\label{eq:REE-hierarchy}
\end{align}
where $I(A:B)$ is the mutual information. 
By construction, the REE satisfies basic properties of an entanglement measure, including the monotonicity under local operations and classical communication (LOCC). 

The relative entropy of entanglement quantifies the information-theoretic distance from $\rho_{AB}$ to the set of separable states. 
The regularized relative entropy of entanglement appears in the generalized quantum Stein's lemma, where it governs the optimal asymptotic rate for distinguishing an entangled state from all separable alternatives (see~\cite{Lami2025Stein} for instance).
Despite its simple definition and conceptual importance, its explicit evaluation is notoriously difficult.
The problem is NP-hard in general~\cite{Huang2014DiscordNPComplete}, and explicit results are known only for special families of states~\cite{Audenaert2002AsymptoticREE, VollbrechtWerner2001Symmetry, HajdusekMurao2013GraphState, MarkhamMiyakeVirmani2007GraphStates, ChenYang2002SchmidtCorrelated, KimHwangJungPark2010REE, TakayanagiUgajinUmemoto2018}.

Here, we propose that the REE admits an intriguing holographic dual.
Consider a boundary holographic state $|\Psi\rangle_{ABC}$ with a tripartition of the boundary into $A$, $B$, and $C$. 
Let $\rho_{AB}$ be the reduced density matrix. 
Our proposal is that 
\begin{align}
E_R(\rho_{A:B}) = 
\frac{1}{4G_N}\mathrm{Area}(\Gamma_{\min})-
\frac{1}{4G_N}\mathrm{Area}(\gamma_{AB})
+O(1).
\label{eq:main-proposal}
\end{align}
Here $\gamma_{AB}$ is the usual minimal surface homologous to the boundary region $AB$, whereas $\Gamma_{\min}$ is the bulk minimal tripartition surface which separates the three boundary regions $A$, $B$, and $C$ (Fig.~\ref{fig_example1}).
Using the Ryu--Takayanagi relation $S(\rho_{AB})=\frac{1}{4G_N}\mathrm{Area}(\gamma_{AB})+O(1)$, our proposal can equivalently be written as
\begin{align}
\min_{\sigma_{AB}\in {\rm SEP}(A:B)} - \Tr \! \left[ \rho_{AB} \log \sigma_{AB} \right] = \frac{1}{4G_N}\mathrm{Area}(\Gamma_{\min}) + O(1). 
\end{align}

\paragraph{Bulk tripartite surface:}
Let us first introduce the notion of bulk tripartite surfaces~\cite{MoriYoshidaConnectedWedge,  GaddeKrishnaSharma2022,
PeningtonWalterWitteveen2023, GaddeKrishnaSharma2023,HarperTakayanagiTsuda2024,
AkersFaulknerLinRath2024, YuanLiZhou2025,IizukaMiyataNishida2025}. 
Given a tripartition of the boundary into $A:B:C$, consider a partition of the bulk into three regions $V_A,V_B,V_C$, containing $A,B,C$ respectively. 
The corresponding bulk tripartition surface $\Gamma$ consists of the interfaces separating these three regions. Namely, it can be decomposed into three portions
\begin{align}
\Gamma = \Gamma_{AB}\cup \Gamma_{BC} \cup \Gamma_{AC}, \qquad \figbox{2.2}{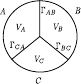}
\end{align}
where, for example, $\Gamma_{AB}$ denotes the interface between $V_A$ and $V_B$.

Among all such bulk tripartitions, we will be particularly interested in the minimal tripartition surface, defined by
\begin{align}
\Gamma_{\min}
\in
\arg\min_{\Gamma}
\mathrm{Area}(\Gamma).
\label{eq:tripartition-min}
\end{align}
Several examples are shown in Fig.~\ref{fig_example1} for AdS$_3$ and tensor networks where $\Gamma_{\min}$ can differ from the usual connected RT surface. 
In particular, even very small tensor networks already admit multiple competing tripartition surfaces, as illustrated by the two- and three-tensor examples.

Tripartite minimal surfaces have been extensively studied in holography
through multipartite replica quantities constructed from permutation
operators~\cite{MoriYoshidaConnectedWedge, GaddeKrishnaSharma2022,
PeningtonWalterWitteveen2023, GaddeKrishnaSharma2023,
HarperTakayanagiTsuda2024, AkersFaulknerLinRath2024,
YuanLiZhou2025,IizukaMiyataNishida2025}.
These quantities involve replica contractions of reduced density
matrices with different permutations, leading to various weighted
multipartite surfaces.
Related geometric quantities, including multipartite volumes,
have also been studied~\cite{FujikiHarperTakayanagiZenoni2026}.

While these constructions provide interesting probes of bulk geometry,
their information-theoretic interpretation often remains unclear.
Unlike genuine entanglement measures, they are not defined through
an operational task or an optimization over a physically distinguished
class of states~\cite{Goto:2026mbj}, and some of these quantities
can fail to be monotonic under partial trace~\cite{HaydenLemmSorce2023}.
Furthermore, taking a von Neumann-like replica limit can be subtle
due to competing saddles and replica-symmetry breaking,
which may change the dominant bulk geometry~\cite{PeningtonWalterWitteveen2023}.
Nevertheless, such quantities have found applications beyond holography,
including proposed probes of the chiral central charge in topologically
ordered phases~\cite{LiuZhangOhyamaKusukiRyu2025,
ShefferFanSternBergRyu2026,GassLevin2026}.

Here, our goal is different.
We seek to relate the multipartite minimal surface to a genuine
entanglement measure whose definition involves an optimization
over separable states.

\begin{figure}
\centering
{
\includegraphics[width=0.9\textwidth]{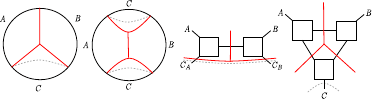}
}
\caption{Tripartite minimal surfaces (red lines) and the Ryu--Takayanagi surfaces for $AB$ (dotted lines). 
}
\label{fig_example1}
\end{figure}

\paragraph{Random tensor networks:}

In this paper, we mainly work with Haar random tensor networks (RTNs) defined on a fixed graph (Fig.~\ref{fig_Haar_RTN}).
Random tensor networks provide an analytically tractable toy model of
holography in which the Ryu--Takayanagi formula and related geometric
structures emerge naturally~\cite{Pastawski:2015qua,Hayden:2016cfa}.
Each bulk vertex carries an independent Haar-random tensor, while neighboring tensors are contracted through maximally entangled bonds. 
The uncontracted legs at the boundary define the Hilbert space of the
boundary state.
RTNs are equipped with natural geometric interpretations. 
For an edge $e$ with bond dimension $D_e$, we can assign the length
\begin{align}
\ell_e := \log_2 D_e ,
\end{align}
so that the area of a bulk cut $\Gamma$ is identified as
\begin{align}
\Area(\Gamma)
=
\sum_{e\in\Gamma}\ell_e .
\end{align}
In this sense, $\log D_e$ plays the role of an elementary area in Planck units, analogous to the factor $\Area/(4G_N)$ in gravity.
We will take the bond dimensions to be exponentially large, $\ell_e=O(n)$, and study the large-$n$ limit.

\paragraph{Upper bound:}
Establishing an upper bound on the REE is comparatively straightforward as it is sufficient to construct a suitable separable state whose relative entropy with respect to $\rho_{AB}$ is small.
Moreover, once such a candidate state is specified, the relative entropy can often be evaluated analytically using standard replica methods.

A useful construction of separable states across $A:B:C$ can be obtained by choosing a tripartition surface $\Gamma$ and dephasing every maximally entangled bond crossing $\Gamma$ in its Schmidt basis:
\begin{align}
\ket{\Phi_e}\bra{\Phi_e}
=
\frac{1}{D_e}
\sum_{i,j=1}^{D_e}
\ket{i,i}\bra{j,j}
\quad\longrightarrow\quad
\omega_e
=
\frac{1}{D_e}
\sum_{i=1}^{D_e}
\ket{i,i}\bra{i,i}.
\label{eq:dephased-EPR}
\end{align}
The tensor network separates into three regions $V_A, V_B, V_C$, and consequently the resulting state $\sigma_{ABC}^{\Gamma}$ is separable across $A:B:C$.

The bond-dephased states $\sigma_{ABC}^{\Gamma}$ provide a natural family of candidate separable states in the optimization for the REE.\footnote{Indeed, the Schmidt-dephased state is the closest separable state for pure bipartite states~\cite{Vedral:1997hd}.} 
Our first general result is that, among this family, the state associated with the minimal tripartition surface $\Gamma_{\mathrm{min}}$ is optimal at leading order. 
More precisely,
\begin{align}
D\!\left(
\rho_{AB}\Vert\sigma^{\Gamma_{\min}}_{AB}
\right)=
\Area(\Gamma_{\min}) - \Area(\gamma_{AB})
+O(1),
\label{eq:general-upper}
\end{align}
and no other tripartition surface \(\Gamma\) gives a smaller relative entropy at leading order. 
This immediately implies
\begin{align}
E_R(\rho_{A:B})
\leq
\Area(\Gamma_{\min}) - \Area(\gamma_{AB})
+O(1).
\end{align}

We establish the upper bound using a replica calculation in which
bond dephasing modifies the usual permutation-spin interactions of RTNs.
The resulting model reduces to an Ising--OR model, where ordinary bonds
carry an Ising domain-wall cost while dephased bonds carry an ``OR''
interaction.
Its ground-state energy is bounded below by the minimal tripartition
area, leading to Eq.~\eqref{eq:general-upper}.

Although derived in RTNs, this construction also suggests a possible
continuum holographic interpretation.
It is natural to ask whether dephasing along $\Gamma$ can be represented
by a bulk defect or boundary condition supported on the tripartition
surface, whose replica action reproduces the corresponding area contribution.
Understanding its gravitational realization, including backreaction
and the appropriate notion of separability in the continuum theory,
remains an open problem.

\paragraph{Lower bound:}

Establishing the matching lower bound is considerably more challenging, as we must rule out the possibility that some other, potentially highly fine-tuned, separable state achieves a substantially smaller relative entropy.
Our main technical tool is a general lower bound in terms of the maximal product-state overlap.
For a tripartite pure state $|\Psi\rangle_{ABC}$, define
\begin{align}
\Lambda(\Psi_{A:B:C})
:=
\max_{\alpha,\beta,\gamma}
\left|
\langle \alpha,\beta,\gamma|\Psi\rangle
\right|^2
\end{align}
where the maximization is over product states across $A:B:C$.
In~\cite{MoriYoshidaREE}, it was shown that
\begin{align}
E_R(\rho_{A:B})
\geq
-S(\rho_{AB})
-
\log_2 \Lambda(\Psi_{A:B:C}).
\end{align}
The maximal product-state overlap is closely related to the geometric measure of entanglement~\cite{WeiGoldbart2003}.

There is a natural class of candidate product states associated with bulk tripartition surfaces.
Given a tripartition surface $\Gamma$, replace every entangled bond crossing $\Gamma$ by a fixed product state, which we take to be $|0\rangle$ without loss of generality.
We denote the resulting boundary product state by
$|\Psi_{ABC}^{\Gamma,|0\rangle}\rangle$.
For the minimal tripartition surface, we find
\begin{align}
\log_2
\left|
\langle
\Psi^{\Gamma_{\min},|0\rangle}
|
\Psi
\rangle
\right|^2
=
-\Area(\Gamma_{\min})+O(1).
\label{eq:geometric-product-overlap}
\end{align}
The remaining question is whether an arbitrary product state can achieve a parametrically larger overlap.
In other words, the lower-bound part of our proposal is reduced to the conjecture
\begin{align}
\log_2 \Lambda(\Psi_{A:B:C})
=
-\Area(\Gamma_{\min})+O(1).
\label{eq:product-overlap-conjecture}
\end{align}
The product state associated with $\Gamma_{\min}$ would then achieve the maximal product overlap at leading order.

\begin{figure}
\centering
{
\includegraphics[width=0.4\textwidth]{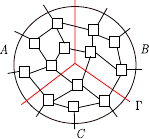}
}
\caption{Haar random tensor network with a bulk tripartition $\Gamma$.}
\label{fig_Haar_RTN}
\end{figure}

In this paper, we establish Eq.~\eqref{eq:product-overlap-conjecture} rigorously for tensor networks consisting of one and two independent Haar-random tensors, as well as three-tensor networks in the Mercedes regime, where the minimal tripartition surface forms a nontrivial tri-junction.

\paragraph{Remark on proof technique}

The main technical ingredient is Gaussian concentration combined with
an $\epsilon$-net counting argument~\cite{HaydenLeungShorWinter2004}.
For a single Haar-random state, one approximates the continuous set
of product states by a finite set of representative states, derives
a concentration bound for each fixed product state, and applies
a union bound over these representatives.
Variants of this approach have been used to determine optimized
entanglement quantities of Haar-random states, including
$E_F$~\cite{Hayden:2005sqo}, $E_R$~\cite{MoriYoshidaREE},
and $E_D^{\mathrm{[LO]}}$~\cite{LiMoriYoshida2026}.
Extending this method to random tensor networks is not straightforward,
since contracting independent Haar-random tensors does not produce
a Haar-random state on the full boundary Hilbert space.

We overcome this difficulty by applying concentration and optimization
sequentially, one tensor at a time.
The resulting estimates are expressed in terms of simpler
product-overlap problems involving the remaining tensors.
For the two-tensor network, for example, the optimization reduces
to bipartite and tripartite product-overlap problems for a single tensor.
Interestingly, these successive steps generate contributions
corresponding to different candidate tripartition surfaces.
In this sense, the proof admits a geometric interpretation
as a local search for the minimal tripartition surface.
The sequential method may also be useful for studying other
optimization-based quantities in random tensor networks.

\paragraph{Hierarchy of geometric transitions}

It is instructive to see how the minimal tripartition surface behaves in a simple continuum geometry.
Consider vacuum AdS$_3$ and two reflection-symmetric boundary intervals $A$ and $B$, each of angular size $\theta$.
The complement $C=C_A\cup C_B$ consists of two equal intervals, each of angular size $\pi-\theta$.
As $\theta$ is increased, the intervals $A$ and $B$ become larger while the regions separating them become smaller.
This one-parameter family provides a simple setting in which the connectivity of the entanglement wedge, the topology of the minimal tripartition surface, and the entanglement-wedge cross section can be compared directly.

Interestingly, these three geometric changes occur at distinct values of $\theta$.
It is useful to emphasize the resulting hierarchy
\begin{align}
\theta_{\rm wedge}
<
\theta_{\Gamma}
<
\theta_{\rm cross},
\qquad
\figbox{1.8}{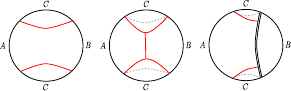}
\label{eq:angles-intro}
\end{align}
The first transition at $\theta=\theta_{\rm wedge}$ is the standard connectivity transition of the entanglement wedge where the mutual information $I(A:B)$ becomes nonzero at leading order.
In the regime $\theta_{\rm wedge}<\theta<\theta_\Gamma$, the minimal tripartition remains the disconnected configuration and our proposal gives $E_R(A:B)=I(A:B)$ at leading order.
At the second transition, $\theta=\theta_\Gamma$, the minimal tripartition changes to the connected two-junction configuration, and the REE departs from the mutual information.
The final transition at $\theta=\theta_{\rm cross}$ is associated with the entanglement-wedge cross section $E_W(A:C)$ which divides the entanglement wedge of $AC$.
(Note that it is between $A$ and $C$, not $A$ and $B$.)
Equivalently, according to~\cite{MoriYoshidaConnectedWedge}, this transition occurs when the locally accessible information $J(A|B) := S_A - E_F(A:C)$ becomes nonzero at leading order.

\paragraph{Plan of the paper}

In Sec.~\ref{sec:replica}, we derive a general upper bound on the REE by evaluating the relative entropy with respect to a bond-dephased separable state.
In Sec.~\ref{sec:maximal-product-overlap}, we develop a lower-bound method based on the maximal product-state overlap and apply it to tensor networks consisting of one, two, and three Haar random tensors.
In Sec.~\ref{sec:hierarchy}, we study the hierarchy of geometric transitions in AdS$_3$.
Finally, in Sec.~\ref{sec:discussion}, we discuss the implications of our results and possible extensions.

\section{Relative entropy for tripartite minimal surface}\label{sec:replica}

Consider a Haar RTN (random tensor network) where physical qubits live on the boundary of the network (Fig.~\ref{fig_Haar_RTN}).
Let $A,B,C$ be boundary regions, which may consist of multiple connected components.
Assume that the RTN graph has graph edge dimensions $D_e$ with
\begin{align}
\ell_e := \log_2 D_e = O(n).
\end{align}
The area of a bulk surface $\Gamma$ is defined as
\begin{align}
\operatorname{Area}(\Gamma)
:=
\sum_{e\in\Gamma} \ell_e
=
\sum_{e\in\Gamma}\log_2 D_e.
\end{align}
We will be interested in the large-$n$ limit.

Let $|\Psi_{ABC}\rangle$ be a normalized pure state generated by the Haar RTN, and let
\begin{align}
\rho_{ABC} = |\Psi_{ABC}\rangle \langle \Psi_{ABC}|.
\end{align}
We define $\sigma_{ABC}^{\Gamma}$ as the normalized mixed state obtained by dephasing every EPR bond cut by the tripartition surface $\Gamma$, while keeping all tensors and all other bonds unchanged.
Since fixing the tripartition surface $\Gamma$ disconnects the three regions, the dephased state $\sigma_{ABC}^{\Gamma}$ is separable across the tripartition $A:B:C$.\footnote{
An EPR edge is given by $\sum_{i,j}|i, i \rangle\langle j,j| $ while a dephased edge is given by $\sum_{i}|i, i \rangle\langle i,i| $ up to normalizations.
}

In this section, we evaluate
\begin{align}
D\!\left(
\rho_{AB}
\middle\|
\sigma_{AB}^{\Gamma}
\right),
\end{align}
where
\begin{align}
\rho_{AB} = \Tr_C \rho_{ABC},
\qquad
\sigma_{AB}^{\Gamma} = \Tr_C \sigma_{ABC}^{\Gamma}.
\end{align}
Our central result is that the relative entropy is minimized, at leading order in the large-$n$ limit, when the dephasing surface $\Gamma$ is chosen to be a minimal bulk tripartition surface $\Gamma_{\min}$.

\begin{lemma}\label{lemma:tripartition}
Among the separable states obtained by dephasing bonds along bulk tripartition surfaces, the minimal tripartition surface $\Gamma_{\min}$ minimizes the relative entropy at leading order:
\begin{align}
\min_{\Gamma}
D\!\left(
\rho_{AB}
\middle\|
\sigma_{AB}^{\Gamma}
\right)
&=
D\!\left(
\rho_{AB}
\middle\|
\sigma_{AB}^{\Gamma_{\min}}
\right)
\nonumber\\
&=
\mathrm{Area}(\Gamma_{\min})
-
\mathrm{Area}(\gamma_{AB})
+
O(1),
\label{eq:tripartition}
\end{align}
where $\gamma_{AB}$ is the minimal-area surface homologous to $AB$.
Consequently, the relative entropy of entanglement satisfies
\begin{align}
E_R(\rho_{AB})
\leq
\mathrm{Area}(\Gamma_{\min})
-
\mathrm{Area}(\gamma_{AB})
+
O(1).
\end{align}
\end{lemma}

Throughout this paper, we assume that there exists a unique minimal tripartition surface.
These results hold with probability approaching one in the large-$n$ limit.

For a non-minimal surface $\Gamma$, the relative entropy is not generally given by $\mathrm{Area}(\Gamma)-\mathrm{Area}(\gamma_{AB})$.
Indeed, the dominant replica configuration can differ from the one aligned with the chosen tripartition surface.
This will become clear from the derivation below.

\paragraph{Normalization.}
The contracted tensor-network state is normalized as
$|\Psi\rangle=|\widetilde\Psi\rangle/\sqrt{N}$,
where $N=\langle\widetilde\Psi|\widetilde\Psi\rangle$ and the deterministic normalization is chosen such that $\mathbb{E}N=1$.
As in standard Haar random tensor networks, normalization fluctuations are exponentially suppressed in the large-bond-dimension limit, with $N=1+O(e^{-cn})$ with high probability for some $c>0$~\cite{Hayden:2016cfa}.
The analogous normalization factor $N_\Gamma$ for the bond-dephased state also concentrates around unity.
These fluctuations do not affect our leading-order results and will be suppressed in what follows.

\subsection{Replica calculation of relative entropy}

We use the replica method to evaluate the relative entropy.
In ordinary Haar random tensor networks, replica averages lead to an effective statistical-mechanical model of permutation-valued spins, with interactions determined by domain walls between neighboring spins.
Our goal is to understand how this picture is modified when the EPR bonds along a bulk tripartition surface $\Gamma$ are dephased.
As we will see, bond dephasing introduces a different type of interaction, which ultimately leads to a geometric description in terms of tripartition surfaces.

Define
\begin{align}
Y_q^{\Gamma}
:=
\Tr\!\left[
\rho_{AB}
\left(\sigma_{AB}^{\Gamma}\right)^{q-1}
\right],
\end{align}
and the ordinary replica partition function
\begin{align}
Z_q := \Tr\!\left[ (\rho_{AB})^q \right].
\end{align}
Using $Z_1=Y_1^{\Gamma}=1$, the relative entropy can be written as
\begin{align}
D\!\left(
\rho_{AB}
\middle\|
\sigma_{AB}^{\Gamma}
\right)
=
\left.
\partial_q
\left[
\log_2 Z_q
-
\log_2 Y_q^{\Gamma}
\right]
\right|_{q=1}.
\label{eq:relative-entropy-replica}
\end{align}

Recall that Haar RTN replica averages produce an $S_q$-valued spin $g_v$ at every bulk tensor, where $S_q$ is the permutation group of $q$ elements~\cite{Hayden:2016cfa}.
For an undephased edge $e=(uv)$, the edge weight is given by
\begin{align}
W_e(g_u,g_v)
=
D_e^{\,\#(g_u^{-1}g_v)-q},
\qquad \figbox{2.5}{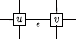}.
\end{align}
Here $\#(g)$ denotes the number of cycles of the permutation $g$.
The corresponding interaction energy is
\begin{align}
-\log_2 W_e(g_u,g_v)
=
\ell_e\, d(g_u,g_v),
\end{align}
where
\begin{align}
d(g,h)
:=
q-\#(g^{-1}h)
\end{align}
is the Cayley distance on $S_q$.
Thus, an ordinary bond carries an energy cost determined by the difference between the permutations at its endpoints.

The boundary conditions for both $Z_q$ and $Y_q^{\Gamma}$ are
\begin{align}
g=e
\qquad &\text{on } C,
\nonumber\\
g=\tau
\qquad &\text{on } A,B,
\end{align}
where $\tau$ is the cyclic permutation.
For the ordinary partition function $Z_q$, the leading contribution is obtained by minimizing the domain-wall energy
\begin{align}
\sum_{e=(uv)} \ell_e\, d(g_u,g_v)
\end{align}
subject to these boundary conditions.
The dominant configuration places $\tau$ inside the bulk region bounded by $AB$ and the minimal surface $\gamma_{AB}$, and $e$ in the complementary region (Fig.~\ref{fig_Ising}).
The resulting domain wall lies along $\gamma_{AB}$, giving
\begin{align}
-\log_2 Z_q
=
(q-1)\,\mathrm{Area}(\gamma_{AB})
+
O(1).
\end{align}

\begin{figure}
\centering
{
\includegraphics[width=0.35\textwidth]{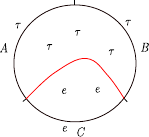}
}
\caption{The leading Ising domain wall contribution to the ordinary replica partition function $Z_q$.}
\label{fig_Ising}
\end{figure}

We now turn to $Y_q^{\Gamma}$, where bonds crossing $\Gamma$ are dephased.
The essential difference is that replica $1$ comes from $\rho$, while replicas $2,\ldots,q$ come from the dephased state $\sigma^{\Gamma}$.
Dephasing removes the off-diagonal coherence between different Schmidt labels on these bonds, imposing additional constraints on the replica contractions.

Consider a dephased edge $e\in\Gamma$.
Its normalized replica weight can be written as
\begin{align}
\widetilde{W}_e(g,h)
=
D_e^{\,\kappa(g,h)-q},
\end{align}
where $\kappa(g,h)$ counts the number of independent bond indices remaining after imposing the dephasing constraints and the replica permutations $g,h\in S_q$ at the two endpoints.

To determine $\kappa(g,h)$, let $i_a$ and $j_a$ denote the ket and bra bond indices in replica $a=1,\ldots,q$.
Since replicas $2,\ldots,q$ come from the dephased state, they obey
\begin{align}
i_a=j_a,
\qquad
a=2,\ldots,q.
\end{align}
The replica permutations at the two endpoints impose
\begin{align}
j_a=i_{g(a)},
\qquad
j_a=i_{h(a)}.
\end{align}
Combining these relations, for $a=2,\ldots,q$ we obtain
\begin{align}
i_a=i_{g(a)}=i_{h(a)},
\end{align}
while for replica $1$ we obtain
\begin{align}
i_{g(1)}=i_{h(1)}.
\end{align}
We define $\kappa(g,h)$ as the number of independent equivalence classes among $i_1,\ldots,i_q$ generated by these relations.
Each independent class contributes a factor $D_e$ from the summation over bond indices.

Analogously to the ordinary Cayley distance, we define the dephased-edge distance
\begin{align}
\delta(g,h)
:=
q-\kappa(g,h),
\end{align}
so that
\begin{align}
-\log_2 \widetilde{W}_e(g,h)
=
\ell_e\,\delta(g,h).
\end{align}
For the identity permutation $e$ and the cyclic permutation $\tau$, one finds
\begin{align}
\delta(e,e)&=0,
\\
\delta(e,\tau)
&=
\delta(\tau,e)
=
\delta(\tau,\tau)
=
q-1.
\end{align}

The last relation reveals an important difference from ordinary replica interactions.
For an undephased bond, identical permutations at the two endpoints carry no energy cost, $d(\tau,\tau)=0$.
By contrast, a dephased bond carries a nonzero cost even when both endpoints have the same cyclic permutation, $\delta(\tau,\tau)=q-1$.
This additional cost arises from the loss of coherence across the dephased bond, rather than from a domain wall between different replica permutations.
As we show in the next subsection, the resulting interaction leads to an effective Ising--OR model, whose ground-state energy is governed by the minimal bulk tripartition surface.

\subsection{Ising--OR model}

We now consider the minimization of the replica action for $Y_q^{\Gamma}$.
For a given configuration of replica spins $\{g_v\}$ with $g_v\in S_q$, the leading large-$n$ action is
\begin{align}
I_{\Gamma}[\{g_v\}]
=
\sum_{e=(uv)\notin\Gamma}
\ell_e\, d(g_u,g_v)
+
\sum_{e=(uv)\in\Gamma}
\ell_e\, \delta(g_u,g_v),
\label{eq:replica-action}
\end{align}
with boundary conditions
\begin{align}
g=e \qquad &\text{on } C,
\nonumber\\
g=\tau \qquad &\text{on } A,B.
\end{align}
Here, $d(g,h)$ is the usual Cayley distance associated with an undephased EPR bond, while $\delta(g,h)$ is the dephased-edge distance introduced above.

A priori, the minimization of Eq.~\eqref{eq:replica-action} must be performed over arbitrary $S_q$-valued spin configurations throughout the bulk.
Remarkably, however, it is sufficient to consider only the identity permutation $e$ and the cyclic permutation $\tau$.

\begin{lemma}\label{lemma:Ising}
The minimum of the replica action is attained by a configuration with
$g_v\in\{e,\tau\}$ at every bulk vertex. Namely,
\begin{align}
\min_{\{g_v\in S_q\}}
I_{\Gamma}[\{g_v\}]
=
\min_{\{g_v\in\{e,\tau\}\}}
I_{\Gamma}[\{g_v\}].
\label{eq:binary-reduction}
\end{align}
\end{lemma}

The key observation is that arbitrary permutation configurations can be bounded by an effective model involving only their distances from the identity.
The resulting energy can always be minimized by configurations taking only the two extreme values, corresponding to $e$ and $\tau$.
We defer the proof of Lemma~\ref{lemma:Ising} to Appendix~\ref{app:binary-reduction}.

It is therefore useful to introduce an Ising-valued variable
\begin{align}
s_v
=
\begin{cases}
0, & g_v=e,\\
1, & g_v=\tau.
\end{cases}
\end{align}
The boundary conditions become
\begin{align}
s=0 \qquad &\text{on } C,
\nonumber\\
s=1 \qquad &\text{on } A,B.
\end{align}
For an undephased edge, we have
\begin{align}
d(g_u,g_v)
=
(q-1)|s_u-s_v|,
\end{align}
while for a dephased edge,
\begin{align}
\delta(g_u,g_v)
=
(q-1)\max\{s_u,s_v\}.
\end{align}
Hence the replica action reduces to
\begin{align}
I_{\Gamma}[s]
=
(q-1) F_\Gamma[s],\quad
F_\Gamma[s]:=
\left[
\sum_{e=(uv)\notin\Gamma}
\ell_e\,|s_u-s_v|
+
\sum_{e=(uv)\in\Gamma}
\ell_e\,\max\{s_u,s_v\}
\right].
\label{eq:Ising-OR-action}
\end{align}

Eq.~\eqref{eq:Ising-OR-action} defines an effective binary spin
model which we refer to as the \emph{Ising--OR model}.
On an undephased edge,
\begin{align}
|s_u-s_v|=s_u\oplus s_v,
\end{align}
so the interaction is the usual Ising domain-wall energy.
On a dephased edge, on the other hand,
\begin{align}
\max\{s_u,s_v\}=s_u\lor s_v,
\end{align}
and the interaction is an OR energy that is nonzero whenever at least one endpoint has $s=1$.
The two types of interactions are illustrated in Fig.~\ref{fig_Ising_OR}.

\begin{figure}
\centering
{
\includegraphics[width=0.65\textwidth]{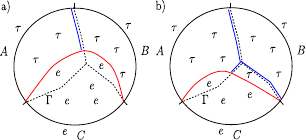}
}
\caption{Leading contributions to the dephased replica partition function $Y_q^{\Gamma}$.
The usual Ising domain wall between the $e$ and $\tau$ regions is shown in red.
Edges on the tripartition surface $\Gamma$ give an additional contribution for the configurations $(e,\tau)$, $(\tau,e)$, and $(\tau,\tau)$, shown in blue.
The tripartition surface $\Gamma$ is indicated by the dotted line.
}
\label{fig_Ising_OR}
\end{figure}

Defining
\begin{align}
\mathcal{T}(\Gamma)
:=
\min_{\substack{
s_v\in\{0,1\}\\
s|_A=s|_B=1,\,
s|_C=0
}}
\left[
\sum_{e=(uv)\notin\Gamma}
\ell_e\,|s_u-s_v|
+
\sum_{e=(uv)\in\Gamma}
\ell_e\,\max\{s_u,s_v\}
\right],
\label{eq:T-Gamma}
\end{align}
we obtain
\begin{align}
-\log_2 Y_q^{\Gamma}
=
(q-1)\mathcal{T}(\Gamma)
+
O(1).
\label{eq:Yq-Ising-OR}
\end{align}

\subsection{Minimal tripartition surface}

The problem of evaluating the mixed replica partition function is thus reduced to finding the ground-state energy of the Ising--OR model.
We now compare the ground-state energies associated with different choices of the tripartition surface.
Our goal is to show that, for any choice of $\Gamma$,
\begin{align}
\mathcal{T}(\Gamma)
\geq
\mathrm{Area}(\Gamma_{\min}),
\end{align}
with equality when $\Gamma=\Gamma_{\min}$.

Before giving the general argument, it is useful to develop some geometric intuition from simple examples.
Consider three neighboring boundary regions $A$, $B$, and $C$, as illustrated in Fig.~\ref{fig_Ising_OR}, and choose a possibly nonminimal tripartition surface $\Gamma$.
The edges of $\Gamma$ carry the OR interaction, while ordinary edges away from $\Gamma$ carry the usual Ising domain-wall interaction.
For a given binary spin configuration, the energy is therefore supported on two types of edges: portions of the defect surface $\Gamma$ on which the OR interaction is nonzero, and ordinary Ising domain walls separating the $e$- and $\tau$-spin regions.

Let us first consider a configuration in which the $e$-spin region extends beyond the bulk region associated with $C$, as shown in Fig.~\ref{fig_Ising_OR}(a).
Along portions of $\Gamma$ whose endpoints carry $(\tau,\tau)$, the OR interaction contributes the full edge weight.
By contrast, a portion of $\Gamma$ whose two endpoints carry $(e,e)$ contributes no OR energy.
Such a missing portion of the defect surface is compensated by an ordinary Ising domain wall separating the enlarged $e$-spin region from the surrounding $\tau$-spin region.
Geometrically, one can ``reroute'' the original tripartition surface through this Ising domain wall.
A subset of the resulting energy-carrying edges therefore forms another tripartition surface $\Gamma'$, whose area cannot exceed the total Ising--OR energy.

A slightly less trivial situation arises when the $e$-spin region crosses the defect surface, as in Fig.~\ref{fig_Ising_OR}(b).
In this case, the union of all energy-carrying defect edges and Ising domain walls need not itself form a tripartition surface.
Nevertheless, one can discard some of these edges and retain a subset which forms a valid tripartition.
The portions of $\Gamma$ that cease to contribute are again replaced by appropriate pieces of the Ising domain wall.
Since all edge weights are positive, the resulting tripartition surface has an area no greater than the total Ising--OR energy.

These examples suggest that the energy-carrying edges always contain a valid tripartition surface.
We can establish this observation directly without explicitly constructing the rerouted surface.

For an arbitrary binary spin configuration $s_v\in\{0,1\}$, define the set of energy-carrying edges
\begin{align}
E_s
:=
&
\left\{
e=(uv)\notin\Gamma:
s_u\neq s_v
\right\}
\nonumber\\
&\quad\cup
\left\{
e=(uv)\in\Gamma:
\max\{s_u,s_v\}=1
\right\}.
\end{align}
By construction, the Ising--OR energy is
\begin{align}
F_\Gamma[s]
=
\sum_{e\in E_s}\ell_e.
\end{align}
The key observation is that removing all edges in $E_s$ necessarily separates the three boundary regions $A$, $B$, and $C$.
To see this, first consider a path connecting $A$ or $B$ to $C$.
Since $s=1$ on $A,B$ and $s=0$ on $C$, any such path must contain an edge across which the spin changes.
Every such edge belongs to $E_s$, regardless of whether it lies on $\Gamma$.
Next, consider a path connecting $A$ to $B$.
If the path avoids all edges in $E_s$, its spin must remain $s=1$ throughout.
Moreover, it cannot cross any edge of $\Gamma$, since every dephased edge with at least one endpoint carrying $s=1$ belongs to $E_s$.
However, $\Gamma$ is a tripartition surface separating $A$ and $B$, so such a path cannot exist.

Therefore, removing $E_s$ separates $A$, $B$, and $C$.
The set $E_s$ may contain redundant edges, but a subset of these edges can be chosen to form a valid tripartition surface $\widetilde{\Gamma}_s$.
Since all edge weights are positive,
\begin{align}
F_\Gamma[s]
=
\sum_{e\in E_s}\ell_e
\geq
\mathrm{Area}(\widetilde{\Gamma}_s)
\geq
\mathrm{Area}(\Gamma_{\min}).
\end{align}
Since this holds for every spin configuration, minimizing over $s$ gives
\begin{align}
\mathcal{T}(\Gamma)
\geq
\mathrm{Area}(\Gamma_{\min})
\label{eq:T-lower-general}
\end{align}
for any choice of the tripartition surface $\Gamma$.

For $\Gamma=\Gamma_{\min}$, the bound is saturated.
Indeed, choose $s_v=1$ on $V_A\cup V_B$ and $s_v=0$ on $V_C$, where $V_A,V_B,V_C$ are the bulk regions separated by $\Gamma_{\min}$.
All ordinary edges then have equal spins at their endpoints and contribute zero, while every edge of $\Gamma_{\min}$ contributes its full weight through the OR interaction.
Hence,
\begin{align}
\mathcal{T}(\Gamma_{\min})
\leq
\mathrm{Area}(\Gamma_{\min}).
\end{align}
Combining this with Eq.~\eqref{eq:T-lower-general}, we obtain
\begin{align}
\mathcal{T}(\Gamma_{\min})
=
\mathrm{Area}(\Gamma_{\min}).
\end{align}
This establishes the main claim of Lemma~\ref{lemma:tripartition}.

The replica calculation above was performed for integer $q\geq 2$,
whereas the relative entropy requires the $q\to 1$ limit of
Eq.~\eqref{eq:relative-entropy-replica}.
Strictly speaking, one must therefore justify this limit.
As shown in Appendix~\ref{app:q1-limit}, the leading-order result
can be established without assuming a smooth replica continuation.
The $q=2$ calculation provides a lower bound on
$-\Tr[\rho_{AB}\log_2\sigma_{AB}^{\Gamma}]$
for any $\Gamma$, while a direct operator inequality gives the
matching upper bound for $\Gamma=\Gamma_{\min}$.
Together, these bounds establish Lemma~\ref{lemma:tripartition}.

\section{Maximal product overlap in random tensor networks}
\label{sec:maximal-product-overlap}

The previous section established the proposed expression for the REE as a
general upper bound for Haar random tensor networks.
To prove the matching lower bound, however, one must rule out the possibility
that an arbitrary separable state achieves a parametrically smaller relative
entropy.

\subsection{Lower bound from product-state overlap}
\label{sec:bound}

A useful approach is to relate the REE to the maximal product-state overlap~\cite{MoriYoshidaREE}.
For a tripartite pure state $\ket{\Psi}_{ABC}$, define
\begin{align}
\Lambda(\Psi_{A:B:C})
:=
\max_{\alpha,\beta,\gamma}
\left|
\langle \alpha,\beta,\gamma|\Psi\rangle
\right|^2.
\end{align}
Equivalently, this can be expressed in terms of the reduced density matrix
$\rho_{AB}=\Tr_C\ket{\Psi}\bra{\Psi}$ as
\begin{align}
\Lambda(\Psi_{A:B:C})
=
\Lambda(\rho_{A:B})
:=
\max_{\alpha,\beta}
\langle\alpha,\beta|\rho_{AB}|\alpha,\beta\rangle.
\label{eq:pure-mixed-product-overlap}
\end{align}
The REE satisfies the general lower bound
\begin{align}
E_R(\rho_{A:B})
\geq
-S(\rho_{AB})
-\log_2\Lambda(\Psi_{A:B:C}).
\label{eq:REE-product-overlap-bound}
\end{align}
This inequality follows from the concavity of the logarithm and the fact
that any separable state is a convex mixture of product states.
For completeness, we provide the proof in
Appendix~\ref{app:product-overlap-proofs}.

We now consider a natural class of product states associated with bulk
tripartition surfaces.
For a given tripartition surface $\Gamma$, replace every EPR bond crossing
$\Gamma$ with a fixed product state $\ket{0}_{\Gamma}$.
We denote the resulting normalized boundary state by
$\ket{\Psi^{\Gamma,|0\rangle}_{ABC}}$.
Since fixing the bond indices along $\Gamma$ disconnects the three bulk
regions, the resulting state factorizes across $A:B:C$,
\begin{align}
\ket{\Psi^{\Gamma,|0\rangle}_{ABC}}
=
\ket{\phi_A^\Gamma}
\otimes
\ket{\phi_B^\Gamma}
\otimes
\ket{\phi_C^\Gamma}.
\label{eq:product-state-factorization}
\end{align}

This construction is closely related to the bond dephasing considered
in Sec.~\ref{sec:replica}.
There, we obtained a separable mixed state by dephasing the EPR bonds
across $\Gamma$.
Here, we instead fix the bond indices to obtain a pure product state.
In both cases, the overlap or relative entropy can be evaluated using
a replica calculation, leading to an effective Ising--OR model.

In particular, the overlap with the geometrically constructed product
state is controlled by the minimal tripartition surface.
We obtain the following result.

\begin{lemma}
For any bulk tripartition surface $\Gamma$,
\begin{align}
-\log_2
\left|
\langle \Psi^{\Gamma,|0\rangle}_{ABC}|\Psi_{ABC}\rangle
\right|^2
\geq
-\log_2
\left|
\langle
\Psi^{\Gamma_{\min},|0\rangle}_{ABC}
|
\Psi_{ABC}
\rangle
\right|^2
+O(1).
\label{eq:product-overlap-minimization}
\end{align}
Moreover,
\begin{align}
-\log_2
\left|
\langle
\Psi^{\Gamma_{\min},|0\rangle}_{ABC}
|
\Psi_{ABC}
\rangle
\right|^2
=
\operatorname{Area}(\Gamma_{\min})+O(1).
\label{eq:product-overlap-minimal-surface}
\end{align}
\end{lemma}

The proof is given in Appendix~\ref{app:product-overlap-proofs}.
It uses a two-replica calculation analogous to that of
Sec.~\ref{sec:replica}, with the cyclic permutation imposed
on all three boundary regions $A$, $B$, and $C$.

Thus, the minimal tripartition surface provides a product state whose
overlap scales as
$2^{-\operatorname{Area}(\Gamma_{\min})}$ at leading order.
However, this does not yet establish the desired lower bound on the REE.
An arbitrary product state, not necessarily associated with a bulk
tripartition surface, might achieve a parametrically larger overlap.

The remaining question is therefore whether the geometrically constructed
product state achieves the maximal possible overlap at leading order,
\begin{align}
-\log_2\Lambda(\Psi_{A:B:C})
\overset{?}{=}
\operatorname{Area}(\Gamma_{\min})+O(1).
\label{eq:maximal-product-overlap-conjecture}
\end{align}
If this relation holds, the general lower bound
Eq.~\eqref{eq:REE-product-overlap-bound}, together with the upper bound
established in Sec.~\ref{sec:replica}, gives
\begin{align}
E_R(\rho_{AB})
=
\operatorname{Area}(\Gamma_{\min})
-
\operatorname{Area}(\gamma_{AB})
+O(1).
\label{eq:REE-conjectured-form}
\end{align}

Hence, establishing our proposal for the REE in Haar random tensor networks
reduces to controlling the maximal overlap with arbitrary tripartite
product states.
In the following subsections, we establish
Eq.~\eqref{eq:maximal-product-overlap-conjecture} for networks consisting
of one, two, and three Haar-random tensors.

\subsection{One random tensor}
\label{sec:one-random-tensor}

We first recall the result for the relative entropy
of entanglement $E_R$ of one Haar random tensor~\cite{MoriYoshidaREE}
and discuss its interpretation in terms of the minimal
tripartition surface.

Let
\begin{align}
\dim\mathcal H_A=a,
\qquad
\dim\mathcal H_B=b,
\qquad
\dim\mathcal H_C=c,
\end{align}
and introduce the logarithmic dimensions (numbers of qubits)
\begin{align}
A=\log_2a,
\qquad
B=\log_2b,
\qquad
C=\log_2c.
\end{align}
Throughout this paper, $K,\eta>0$ denote universal constants
whose values may change from line to line.

We represent a Haar random state $|\Psi\rangle_{ABC}$ using
a complex Gaussian tensor
\begin{align}
|G\rangle_{ABC}
=
\sum_{i=1}^{a}
\sum_{j=1}^{b}
\sum_{k=1}^{c}
G_{ijk}|i,j,k\rangle
=
\figbox{2.0}{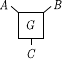}
\end{align}
with covariance
\begin{align}
\mathbb E
\left[
G_{ijk}\overline{G_{i'j'k'}}
\right]
=
\frac{1}{abc}\,
\delta_{ii'}\delta_{jj'}\delta_{kk'}.
\label{eq:one-tensor-Gaussian-covariance}
\end{align}
We assume that the logarithmic dimensions obey the
triangle inequalities,
\begin{align}
A<B+C,
\qquad
B<C+A,
\qquad
C<A+B.
\end{align}
These conditions exclude trivial regimes where one subsystem
is larger than the other two combined and the state can be
approximately reduced to EPR pairs by a local unitary.

The normalized Haar random state is
\begin{align}
|\Psi\rangle_{ABC}
=
\frac{|G\rangle_{ABC}}{\sqrt{\langle G|G\rangle}}.
\label{eq:one-tensor-Haar-state}
\end{align}
Since $\langle G|G\rangle$ concentrates sharply around unity
in the large-dimension limit, we will suppress the
normalization factor in the following calculations.
The normalization fluctuations are discussed in
Appendix~\ref{app:one-tensor-overlap}.

The main result is the following bound on the maximal
product-state overlap.

\begin{lemma}\label{lem:overlap}
There exist universal constants $K,\eta>0$ such that
\begin{align}
\Pr\left[
\Lambda(\rho_{A:B})
>
K
\left(
\frac{1}{ab}
+
\frac{1}{ac}
+
\frac{1}{bc}
\right)
\right]
\leq
e^{-\eta(a+b+c)}.
\label{eq:one-tensor-product-overlap-probability}
\end{align}
\end{lemma}

Together with the geometric product-state construction
of Sec.~\ref{sec:bound}, the bound implies
\begin{align}
-\log_2\Lambda(\rho_{A:B})
=
\min\{A+B,A+C,B+C\}
+
O(1)
\label{eq:one-tensor-product-overlap-area}
\end{align}
with high probability.
For one Haar random tensor, the candidate minimal
tripartition surfaces are
\begin{align}
\figbox{2.0}{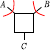},
\quad
\figbox{2.0}{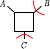},
\quad
\figbox{2.0}{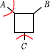}.
\end{align}
Their areas are $A+B$, $B+C$, and $C+A$, respectively.
Thus, the three terms appearing in the product-overlap
bound correspond precisely to the three candidate
tripartition surfaces, giving
\begin{align}
-\log_2\Lambda(\rho_{A:B})
=
\mathrm{Area}(\Gamma_{\min})+O(1).
\end{align}

Since the triangle inequalities imply
$S(\rho_{AB})=C+O(1)$, the lower bound
Eq.~\eqref{eq:REE-product-overlap-bound} gives
\begin{align}
E_R(\rho_{AB})
\geq
\min\{A+B,A+C,B+C\}
-
C
+
O(1).
\end{align}
Combining this with the upper bound from
Sec.~\ref{sec:replica}, we obtain
\begin{align}
E_R(\rho_{AB})
&=
\min\{A+B,A+C,B+C\}-C+O(1)
\nonumber\\
&=
\mathrm{Area}(\Gamma_{\min})
-
\mathrm{Area}(\gamma_{AB})
+
O(1).
\label{eq:one-tensor-REE-result}
\end{align}
This confirms the proposed relation between the REE
and the minimal tripartition surface for one Haar
random tensor.

\paragraph{Sketch of proof.}
The proof uses Gaussian concentration and $\epsilon$-net
arguments, but the underlying idea can be understood as a
simple counting argument.
For fixed normalized states $|\alpha\rangle$ and $|\beta\rangle$,
consider the projected Gaussian vector
\begin{align}
|G_{\alpha,\beta}\rangle_C
:=
{}_{AB}\langle\alpha,\beta|G\rangle_{ABC}.
\end{align}
Its squared norm has expectation
$\mathbb E\|G_{\alpha,\beta}\|^2=1/(ab)$,
so its typical norm is of order $1/\sqrt{ab}$.
Moreover, Gaussian concentration implies
\begin{align}
\Pr\left[
\|G_{\alpha,\beta}\|
>
\frac{1}{\sqrt{ab}}
+
\frac{t}{\sqrt{abc}}
\right]
\leq
e^{-\eta t^2}.
\end{align}
Thus, the probability of a large deviation from the
typical value is exponentially suppressed.

So far, the bound applies to fixed $\alpha$ and $\beta$.
To bound the maximal product overlap, however, we must
optimize over all possible choices of these states.
Although there are continuously many such states,
they can be approximated, to a fixed accuracy,
by a finite set of representative states.
For a Hilbert space of dimension $a$, the number of
representative states can be chosen to be at most $e^{Ka}$.
This is the basic idea of an $\epsilon$-net.
Similarly, we need at most $e^{Kb}$ representative
states for $\beta$.

We can therefore apply the concentration bound to
all representative states simultaneously using a union bound.
The exponentially large number of choices is compensated
by choosing the deviation parameter $t$ of order
$\sqrt{a+b+c}$.
In particular, optimizing over $\beta$ introduces
an additional contribution of order
$\sqrt{b}/\sqrt{abc}=1/\sqrt{ac}$,
while optimizing over $\alpha$ introduces a contribution
of order $\sqrt{a}/\sqrt{abc}=1/\sqrt{bc}$.
The resulting bound is
\begin{align}
\max_{\alpha,\beta}\|G_{\alpha,\beta}\|
\leq
K\left(
\frac{1}{\sqrt{ab}}
+
\frac{1}{\sqrt{ac}}
+
\frac{1}{\sqrt{bc}}
\right)
\end{align}
with exponentially high probability.
Squaring this estimate gives the claimed bound
on the maximal product overlap.
The complete proof is given in
Appendix~\ref{app:one-tensor-overlap}.

\subsection{Two Haar random tensors}
\label{sec:two-random-tensors}

We next derive the relative entropy of entanglement $E_R$
for a tensor network consisting of two independent Haar random tensors
connected through an internal bond $R$.

Consider two tensors with external legs $A,C_A$ and $B,C_B$,
respectively, connected by an internal bond $R$.
We denote their Hilbert space dimensions by
\begin{align}
a=2^A,\qquad
b=2^B,\qquad
c_A=2^{C_A},\qquad
c_B=2^{C_B},\qquad
r=2^R.
\end{align}
As in the one-tensor case, we represent the two independent
Haar random tensors using complex Gaussian tensors
\begin{align}
|G\rangle_{A C_A R}
&=
\sum_{i=1}^{a}
\sum_{\mu=1}^{c_A}
\sum_{\ell=1}^{r}
G_{i\mu\ell}|i,\mu,\ell\rangle,
\\
|H\rangle_{B C_B R}
&=
\sum_{j=1}^{b}
\sum_{\nu=1}^{c_B}
\sum_{\ell=1}^{r}
H_{j\nu\ell}|j,\nu,\ell\rangle,
\end{align}
with covariances
\begin{align}
\mathbb E\left[
G_{i\mu\ell}\overline{G_{i'\mu'\ell'}}
\right]
&=
\frac{1}{a c_A r}
\delta_{ii'}\delta_{\mu\mu'}\delta_{\ell\ell'},
\label{eq:two-tensor-G-covariance}
\\
\mathbb E\left[
H_{j\nu\ell}\overline{H_{j'\nu'\ell'}}
\right]
&=
\frac{1}{b c_B r}
\delta_{jj'}\delta_{\nu\nu'}\delta_{\ell\ell'}.
\label{eq:two-tensor-H-covariance}
\end{align}
We assume that the logarithmic dimensions of each local tensor
obey the triangle inequalities,
\begin{align}
A<C_A+R,\quad
C_A<A+R,\quad
R<A+C_A,
\label{eq:two-tensor-G-triangle}
\end{align}
and
\begin{align}
B<C_B+R,\quad
C_B<B+R,\quad
R<B+C_B.
\label{eq:two-tensor-H-triangle}
\end{align}

We use the symbol $\star$ to denote contraction over the
common $R$ leg,
\begin{align}
|G\star H\rangle_{A C_A B C_B}
:=
\sum_{\ell=1}^{r}
{}_R\langle\ell|G\rangle_{A C_A R}
\otimes
{}_R\langle\ell|H\rangle_{B C_B R}
=
\figbox{2.0}{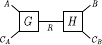}.
\label{eq:two-tensor-star-definition}
\end{align}
Although both Gaussian tensors have unit norm on average,
their contraction does not.
In fact,
\begin{align}
\mathbb E\langle G\star H|G\star H\rangle
=
\frac{1}{r}.
\end{align}
The normalized global state is therefore
\begin{align}
|\Psi\rangle_{A C_A B C_B}
\simeq
\sqrt r\,|G\star H\rangle_{A C_A B C_B},
\end{align}
where the factor $\sqrt r$ compensates for the contraction
over the internal bond.
As before, normalization fluctuations are suppressed
in the large-dimension limit.

For normalized states $|\alpha\rangle\in\mathcal H_A$
and $|\beta\rangle\in\mathcal H_B$, let
\begin{align}
|G_\alpha\rangle
&:=
{}_A\langle\alpha|G\rangle,
\\
|H_\beta\rangle
&:=
{}_B\langle\beta|H\rangle.
\end{align}
The maximal product overlap can then be expressed as
\begin{align}
\Lambda(\rho_{A:B})
&=
\max_{\alpha,\beta}
\langle\alpha,\beta|
\rho_{AB}
|\alpha,\beta\rangle
\nonumber\\
&=
r\max_{\alpha,\beta}
\bigl\|\,|G_\alpha\star H_\beta\rangle\,\bigr\|^2,
\label{eq:two-tensor-maximal-overlap}
\end{align}
where $\star$ denotes contraction over the common $R$ leg.

Our goal is to show that the maximal product overlap
is controlled by the minimal tripartition surface.
Under the local triangle inequalities, there are four
candidate tripartition surfaces,
\begin{align}
\figbox{2.0}{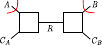},
\quad
\figbox{2.0}{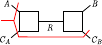},
\quad
\figbox{2.0}{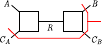},
\quad
\figbox{2.0}{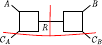}.
\end{align}
Their areas are
\begin{align}
A+B,\quad
A+C_A+C_B,\quad
B+C_A+C_B,\quad
C_A+C_B+R,
\end{align}
respectively.
Hence,
\begin{align}
\mathrm{Area}(\Gamma_{\min})
=
\min\left\{
A+B,\,
A+C_A+C_B,\,
B+C_A+C_B,\,
C_A+C_B+R
\right\}.
\end{align}

\begin{lemma}\label{lemma:2TN}
There exist universal constants $K,\eta>0$ such that
\begin{align}
\Pr\left[
\Lambda(\rho_{A:B})
>
K\left(
\frac{1}{ab}
+
\frac{1}{ac_Ac_B}
+
\frac{1}{bc_Ac_B}
+
\frac{1}{c_Ac_Br}
\right)
\right]
\leq
\exp\left[
-\eta\min\{a,b,c_A,c_B,r\}
\right].
\label{eq:two-tensor-product-overlap-probability}
\end{align}
\end{lemma}

Together with the geometric product-state construction
of Sec.~\ref{sec:bound}, this bound implies
\begin{align}
-\log_2\Lambda(\rho_{A:B})
=
\mathrm{Area}(\Gamma_{\min})+O(1)
\label{eq:two-tensor-product-overlap-area}
\end{align}
with high probability.

Combining Eq.~\eqref{eq:two-tensor-product-overlap-area}
with the lower bound of Sec.~\ref{sec:bound}
and the upper bound of Sec.~\ref{sec:replica},
we obtain
\begin{align}
E_R(\rho_{AB})
=
\mathrm{Area}(\Gamma_{\min})
-
\mathrm{Area}(\gamma_{AB})
+
O(1).
\label{eq:two-tensor-REE-result}
\end{align}
Thus, the proposed geometric relation for the REE
continues to hold for two independently random tensors.

\paragraph{Sketch of proof.}
The proof extends the Gaussian concentration and counting argument
of Sec.~\ref{sec:one-random-tensor} by applying it sequentially
to the two independent random tensors.
We first average over $H$ while keeping $G,\alpha,\beta$ fixed.
The expectation value is controlled by $\|G_\alpha\|$,
while the fluctuations are controlled by another quantity
$\lambda_\alpha(G)$, which measures the maximal product
amplitude of the projected tensor $G_\alpha$.
We then optimize over $\beta$ and $\alpha$ using the same
$\epsilon$-net counting argument as in the one-tensor case.

The key observation is that these two quantities
correspond to simpler overlap problems involving only
the remaining tensor $G$,
\begin{align}
\bigl\|\,|G_\alpha\rangle\,\bigr\|
&=
\max_\phi |\langle\alpha,\phi|G\rangle|
=
\figbox{2.5}{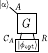},
\\
\lambda_\alpha(G)
&=
\max_{\gamma,u}
|\langle\alpha,\gamma,u|G\rangle|
=
\figbox{2.5}{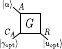}.
\end{align}
The first is a bipartite overlap problem across
$A:(C_A R)$, where $\phi$ is an arbitrary state on $C_A R$,
while the second is a tripartite product-overlap problem
across $A:C_A:R$.
Both can be bounded using the one-tensor results
of Sec.~\ref{sec:one-random-tensor}.

Thus, the sequential optimization reduces the two-tensor problem
to simpler overlap problems for the remaining tensor.
One subtlety is that the resulting bound depends on the order
of optimization.
By also considering the opposite order, averaging over $G$
before $H$, and choosing the better bound, we recover all
four candidate tripartition surfaces.
In this sense, the proof can be viewed as a local search
for the minimal tripartition surface.
The complete proof is given in
Appendix~\ref{app:two-tensor-overlap}.

\subsection{Three Haar random tensors}
\label{sec:three-random-tensors}

Finally, we derive the relative entropy of entanglement
$E_R$ for a tensor network consisting of three independent
Haar random tensors connected in a triangle.
We focus on the Mercedes regime, where the minimal tripartition
surface cuts the three internal bonds of the network.

Consider three tensors with external legs $A$, $B$, and $C$,
connected by internal bonds $R$, $C_A$, and $C_B$.
We denote their Hilbert space dimensions by
\begin{align}
a=2^A,\qquad
b=2^B,\qquad
c=2^C,\qquad
c_A=2^{C_A},\qquad
c_B=2^{C_B},\qquad
r=2^R.
\end{align}
As in the previous subsections, we represent the three independent
Haar random tensors by complex Gaussian tensors
$|G\rangle_{A C_A R}$, $|H\rangle_{B C_B R}$,
and $|J\rangle_{C C_A C_B}$.
For example,
\begin{align}
|G\rangle_{A C_A R}
=
\sum_{i=1}^{a}
\sum_{\mu=1}^{c_A}
\sum_{\ell=1}^{r}
G_{i\mu\ell}|i,\mu,\ell\rangle,
\end{align}
with covariance
\begin{align}
\mathbb E\left[
G_{i\mu\ell}
\overline{G_{i'\mu'\ell'}}
\right]
=
\frac{1}{a c_A r}
\delta_{ii'}\delta_{\mu\mu'}\delta_{\ell\ell'}.
\end{align}
The tensors $H$ and $J$ are defined analogously,
with independent Gaussian entries of variances
$1/(bc_Br)$ and $1/(cc_Ac_B)$, respectively.
All three tensors have unit norm on average.

We assume that the logarithmic dimensions
of each local tensor obey the triangle inequalities,
\begin{align}
A<C_A+R,\quad
C_A<A+R,\quad
R<A+C_A,
\end{align}
\begin{align}
B<C_B+R,\quad
C_B<B+R,\quad
R<B+C_B,
\end{align}
and
\begin{align}
C<C_A+C_B,\quad
C_A<C+C_B,\quad
C_B<C+C_A.
\end{align}

We use the symbol $\star$ to denote contraction
over all common internal legs.
In particular,
\begin{align}
|G\star H\star J\rangle_{ABC}
:=
\sum_{i,j,k,\mu,\nu,\ell}
G_{i\mu\ell}
H_{j\nu\ell}
J_{k\mu\nu}
|i,j,k\rangle.
\label{eq:three-tensor-star-definition}
\end{align}
Although the symbol is written sequentially,
it denotes the simultaneous contraction of
the three internal bonds $R$, $C_A$, and $C_B$.

The contracted Gaussian network has
\begin{align}
\mathbb E
\langle G\star H\star J
|G\star H\star J\rangle
=
\frac{1}{r c_A c_B}.
\end{align}
The normalized global state is therefore
\begin{align}
|\Psi\rangle_{ABC}
\simeq
\sqrt{r c_A c_B}\,
|G\star H\star J\rangle_{ABC}
=
\sqrt{r c_A c_B}\,
\figbox{2.0}{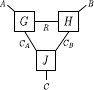}.
\label{eq:three-tensor-global-state}
\end{align}
As in the one- and two-tensor cases,
normalization fluctuations are strongly suppressed
in the large-dimension limit.

For normalized states
$|\alpha\rangle\in\mathcal H_A$ and
$|\beta\rangle\in\mathcal H_B$, define
\begin{align}
|G_\alpha\rangle
&:=
{}_A\langle\alpha|G\rangle,
\\
|H_\beta\rangle
&:=
{}_B\langle\beta|H\rangle.
\end{align}
The maximal product overlap is
\begin{align}
\Lambda(\rho_{A:B})
&=
\max_{\alpha,\beta}
\langle\alpha,\beta|
\rho_{AB}
|\alpha,\beta\rangle
\nonumber\\
&=
r c_A c_B
\max_{\alpha,\beta}
\bigl\|\,
|G_\alpha\star H_\beta\star J\rangle
\,\bigr\|^2.
\label{eq:three-tensor-Lambda}
\end{align}
Here, the projected tensor
$|G_\alpha\star H_\beta\star J\rangle$
is obtained by contracting the boundary
legs $A$ and $B$ with $\alpha$ and $\beta$,
respectively.

Under the local triangle inequalities,
there are four potentially minimal
tripartition surfaces,
\begin{align}
\figbox{2.0}{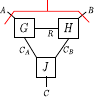},
\quad
\figbox{2.0}{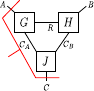},
\quad
\figbox{2.0}{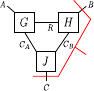},
\quad
\figbox{2.0}{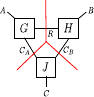}.
\end{align}
Their areas are
\begin{align}
A+B,\quad
A+C,\quad
B+C,\quad
R+C_A+C_B,
\end{align}
respectively.

In this subsection, we focus on the Mercedes phase,
where the internal tripartition surface
has the smallest area,
\begin{align}
R+C_A+C_B
<
\min\{A+B,A+C,B+C\}.
\label{eq:three-tensor-Mercedes-phase}
\end{align}
In this regime,
\begin{align}
\mathrm{Area}(\Gamma_{\min})
=
R+C_A+C_B.
\end{align}

\begin{lemma}\label{lemma:3TN}
Let
\begin{align}
m:=\min\{a,b,c,c_A,c_B,r\}.
\end{align}
There exist universal constants $K,\eta>0$ such that
\begin{align}
\Pr\left[
\Lambda(\rho_{A:B})
>
\frac{K}{r c_A c_B}
\right]
\leq
e^{-\eta m}.
\label{eq:three-tensor-product-overlap-probability}
\end{align}
\end{lemma}

The proof is given in
Appendix~\ref{app:three-tensor-overlap}.
Together with the geometric product-state construction
of Sec.~\ref{sec:bound}, this establishes
\begin{align}
-\log_2\Lambda(\rho_{A:B})
=
R+C_A+C_B+O(1)
=
\mathrm{Area}(\Gamma_{\min})+O(1)
\label{eq:three-tensor-product-overlap-area}
\end{align}
with high probability.

In the Mercedes regime, the minimal RT surface
homologous to $AB$ has area
\begin{align}
\mathrm{Area}(\gamma_{AB})=C.
\end{align}
Combining the product-overlap lower bound
with the general upper bound on the REE,
we obtain
\begin{align}
E_R(\rho_{A:B})
&=
R+C_A+C_B-C+O(1)
\nonumber\\
&=
\mathrm{Area}(\Gamma_{\min})
-
\mathrm{Area}(\gamma_{AB})
+
O(1).
\label{eq:three-tensor-ER}
\end{align}
Thus, the proposed geometric relation for the REE
also holds when the minimal tripartition surface
has a nontrivial Mercedes-type configuration.

\section{Hierarchy of geometric transitions in AdS$_3$}\label{sec:hierarchy}
Let us illustrate rich structures of the connected wedge in AdS/CFT. As we have seen for the relative entropy of entanglement as well as the distillable entanglement~\cite{MoriYoshidaConnectedWedge}, there are several distinct phases even when the entanglement wedge is connected.

For simplicity, we focus on pure global AdS$_3$. Consider the reflection-symmetric boundary ordering $A,\ C_A,\ B,\ C_B$, with $|A|=|B|=\theta$ and $|C_A|=|C_B|$ as shown in Fig.~\ref{fig_hierarchy}.
The AdS radius is rescaled to unity and the total circumference is set to $2\pi$.

\begin{figure}
\centering
{
\includegraphics[width=0.8\textwidth]{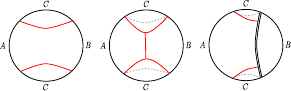}
}
\caption{Transitions for connected entanglement wedge, minimal tripartition surface, and entanglement wedge cross section. Note that the entanglement wedge cross section is considered between $A$ and $C$. 
}
\label{fig_hierarchy}
\end{figure}

To see how $A$ and $B$ are entangled, consider mutual information
\begin{equation}
    I(A:B)=S_A+S_B-S_{AB},
\end{equation}
the holographic one-way distillable entanglement~\cite{MoriYoshidaConnectedWedge}~\footnote{$J_W$ is also dual to the locally accessible information on the boundary, quantifying the information gain by local measurements.}
\begin{equation}
    E_D(A\leftarrow B)=J_W(A|B):=S_A-E_W(A:C_A C_B),
\end{equation}
and relative entropy of entanglement $E_R(A:B)$.
Here $E_W(X:Y)$ denotes the entanglement wedge cross section between $X$ and $Y$, which is dual to the entanglement of formation $E_F$~\cite{MoriYoshidaConnectedWedge} as well as the entanglement of purification $E_P$~\cite{TakayanagiUmemoto2018}.

\paragraph{Mutual information}

The mutual information exhibits a phase transition when $A$ and $B$ each covers a quarter of the full circumference. The leading-order value is
\begin{equation}
    I(A:B) =
    \begin{cases}
        0, \qquad & \displaystyle\theta<\frac{\pi}{2}\approx 1.57\\
        \displaystyle\frac{2c}{3} \log\frac{\sin\displaystyle\frac{\theta}{2}}{\sin\displaystyle\frac{\pi-\theta}{2}}, \qquad & \displaystyle\frac{\pi}{2}\le \theta< \pi
    \end{cases}
\end{equation}
where we used the Brown-Henneaux relation $c=\frac{3}{2G_N}$.

This transition corresponds to the disconnected--connected wedge transition for the entanglement wedge of $\rho_{AB}$. The entanglement wedge cross section $E_W(A:B)$ also exhibits a phase transition from zero to finite value simultaneously.

\paragraph{One-way distillable entanglement}

Since the mutual information only captures the total correlation, let us consider a finer, operational measure: one-way distillable entanglement $E_D(A\leftarrow B)$, which quantifies the maximum number of Bell pairs that can be distilled from $\rho_{AB}$ by one-way LOCC from $B$ to $A$.\footnote{See~\cite{MoriYoshidaConnectedWedge,LiMoriYoshida2026} for the precise definition.}
Holographically, the leading-order value is given by
\begin{equation}
    J_W(A|B) =
    \begin{cases}
        0,\qquad &\theta< \theta_{\rm cross}:=4\, \mathrm{arctan} \sqrt{\sqrt{2}-1}\approx 2.29\\
        \displaystyle\frac{c}{3} \log\frac{\displaystyle\sin\frac{\theta}{2}\tan\frac{\theta}{2}}{2}, \qquad & \theta_{\rm cross}\le \theta< \pi.
    \end{cases}
\end{equation}
The detailed calculation is presented in Appendix A of~\cite{MoriYoshidaConnectedWedge}.

\paragraph{Relative entropy of entanglement}

Finally, the relative entropy of entanglement also exhibits a transition from $O(1)$ to $O(c)$. When $A-B$ wedge is connected, the problem of finding the minimal tripartition surface (Steiner trees) amounts to optimizing over two bulk triway intersection points. The same computation is carried out in~\cite{Anegawa:2025prn, GaddeKrishnaSharma2023} and we find, at the leading order,
\begin{align}
    E_R(A:B)  &=\frac{\mathrm{Area}\,\Gamma_{\min}}{4G_N} -S_{AB}\\
    &=
    \begin{cases}
        0, \qquad & \theta<\displaystyle\frac{\pi}{2}\\
        I(A:B),\qquad & \displaystyle\frac{\pi}{2}\le \theta\le \theta_{\Gamma} \approx 2.06 \\
        -\displaystyle\frac{c}{6}\log \displaystyle\frac{27}{16} + \frac{c}{3}\log\frac{\qty(1+\displaystyle\cos\frac{\pi-\theta}{2})\displaystyle\cos\frac{\theta}{2}}{\displaystyle\sin^2\frac{\pi-\theta}{2}}, \qquad & \theta_{\Gamma}\le \theta <\pi.
    \end{cases}
\end{align}
Here $\Gamma_{\rm min}$ is either the union of the disjoint RT surfaces for $A$ and $B$ or
the minimal tripartition surface like the middle figure in Fig.~\ref{fig_hierarchy}.
$\theta_{\Gamma}$ is determined so that the transition becomes continuous. Its precise value can be found by solving
\begin{equation}
    \sinh \qty[\log\qty(\sqrt{\frac{27}{16}} \tan^2\frac{\theta_{\Gamma}}{2})] = \tan\frac{\theta_{\Gamma}}{2}.
\end{equation}





\paragraph{Transitions of mixed-state wedges}

First, the entanglement wedge of $A\cup B$ becomes connected as we increase $\abs{A}=\abs{B}=\theta$ from $0$.
The transition is determined by equality of the connected and disconnected RT configurations. The transition occurs at $\theta=\pi/2\approx 1.57$ where the mutual information becomes nonzero. At this point, $A-B$ starts to share geometric entanglement.


Second, the minimal tripartition surface undergoes its own connectivity transition.
The closest separable state is no longer given by the disconnected entanglement wedge; instead, it is given by the connected tripartition wedge with two bulk trijunctions.
The transition occurs at $\theta=\theta_{\rm \Gamma}\approx 2.06$. The REE starts to deviate from the mutual information at this point.



Finally, there is a further transition associated with the entanglement wedge cross section $E_W(A:C_AC_B)=E_W(B:C_AC_B)$.
The entanglement wedge after EPR-distilling LOCC becomes connected only after $\theta=\theta_{\rm cross}\approx 2.29$. At this point, a part of geometric entanglement shared between $A$ and $B$ becomes (one-shot) one-way LOCC distillable.

The three transition scales therefore satisfy
\begin{equation}
    \theta_{\rm wedge}:=\frac{\pi}{2} < \theta_{\Gamma}\approx 2.06 < \theta_{\rm cross} \approx 2.29.
\end{equation}
The plot of the transition points and each measure is presented in Fig.~\ref{fig:plot}.

\begin{figure}
    \centering
    \includegraphics[width=0.6\linewidth]{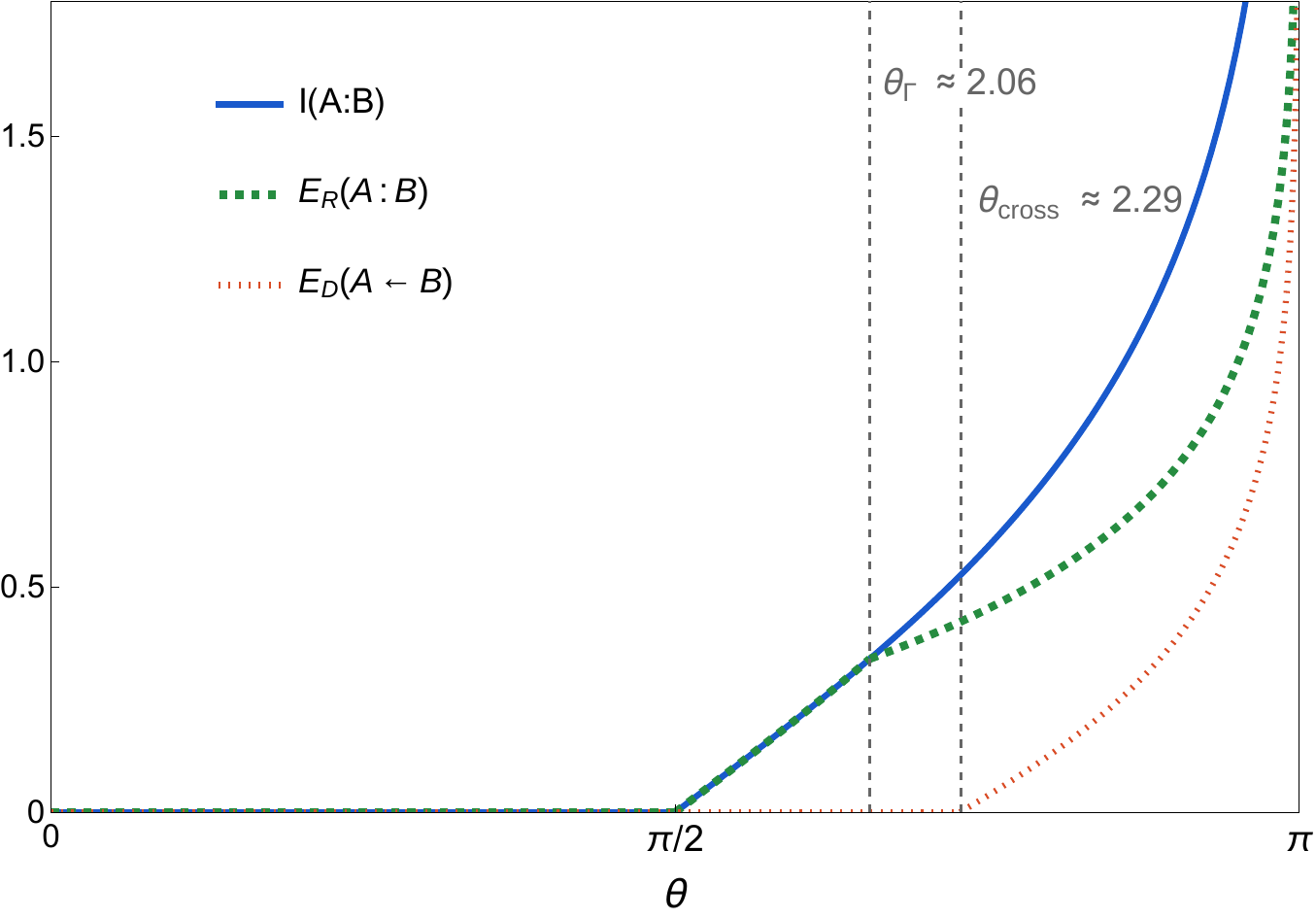}
    \caption{Plot of holographic mixed-state measures. There are distinct phases whose boundaries are specified by the angles $\theta_{\rm wedge},\theta_{\Gamma},\theta_{\rm cross}$.}
    \label{fig:plot}
\end{figure}

\section{Discussion}\label{sec:discussion}

In this paper, we proposed a bulk dual of the relative entropy of entanglement $E_R$ in terms of the bulk minimal tripartition surface. 
By focusing on random tensor networks, we developed a particular recursive method for determining $E_R$ for small networks. 
Despite great interest and importance in quantum information, analytic or numerical computation of $E_R$ still remains challenging. 
Our results may also be of independent interest from a purely quantum-information-theoretic perspective.
We conclude with various discussions. 

\paragraph{Broader motivation:}

It is useful to place the present results in the broader context of multipartite entanglement in holography (see~\cite{AkersRath2020} for instance).
Despite substantial progress, the structure and operational significance of multipartite entanglement in holographic states remain much less understood.
At the same time, genuinely multipartite correlations appear to play an important role in several conceptual questions in quantum gravity, including recent discussions of semiclassical baby universes~\cite{MoriYoshidaBabyUniverse}.
Our results provide a new geometric quantity that is directly tied to such multipartite structure.

It is also notable that the relevant bulk object is not an ordinary minimal surface, but a tripartition surface that can contain a nontrivial tri-junction.
Such junction structures probe the organization of bulk geometry in a more local and fine-grained way than standard bipartite cuts.
It would be interesting to understand whether this perspective can be sharpened into a useful probe of bulk locality, perhaps even at sub-AdS scales.

\paragraph{Multipartite generalization:}

A natural extension of our work is to consider the relative entropy with respect to the set of $m$-partite separable states.
It is tempting to conjecture that the corresponding holographic dual is controlled by a minimal $m$-partition surface in the bulk.
It would be interesting to determine whether the arguments developed here for tripartite separability admit such a generalization.

\paragraph{Genuine multi-entropy:}

It is interesting to relate our proposal to the genuine tripartite multi-entropy ${\rm GM}^{(3)}$ introduced in~\cite{IizukaNishida2025}.
The minimal tripartition surface was previously proposed to be dual to the tripartite multi-entropy,
\begin{align}
S^{(3)}(A:B:C)
=
\frac{\Area(\Gamma_{\min})}{4G_N}.
\end{align}
For a pure state on $ABC$, our proposal therefore implies
\begin{align}
E_R(A:B)-\frac12 I(A:B)
&=
S^{(3)}(A:B:C)
-\frac12\bigl(S_A+S_B+S_C\bigr)
\nonumber\\
&=
{\rm GM}^{(3)}(A:B:C)
\end{align}
at leading order.
The genuine multi-entropy ${\rm GM}^{(3)}$ was shown to be nonnegative for holographic states~\cite{IizukaNishida2025}.
Our proposal therefore implies the simple inequality
\begin{align}
E_R(A:B)\geq \frac12 I(A:B)
\end{align}
at leading order.

\paragraph{On entanglement of formation $E_F$ in the two-tensor model:}

Our results also have immediate implications for the entanglement of
formation $E_F(\rho_{AB})$.
In the holography literature, it has been proposed that, at leading order,
$E_F(\rho_{AB})$ is given by the area of the entanglement wedge cross
section $E_W$.\footnote{
A certain no-go argument against this proposal was given by Umemoto
in~\cite{Umemoto2019}.
We recently pointed out that this argument can be circumvented if the
proposal is understood as an approximate relation holding only at leading
order~\cite{MoriYoshidaConnectedWedge}.
}
For the two-tensor network, let us focus on the regime where the
entanglement wedge is connected and $R$ is the minimal cross section.
Then, this proposal reduces to
\begin{align}
E_F(\rho_{AB}) \overset{?}{=} R + O(1).
\end{align}
Indeed, one can readily construct a candidate decomposition of $\rho_{AB}$
(the two-tensor network itself), giving the upper bound
\begin{align}
E_F(\rho_{AB}) \leq R + O(1).
\end{align}

When
\begin{align}
R < A+B-C_A-C_B,
\end{align}
our result gives $E_R=R+O(1)$.
Together with the general inequality $E_R\leq E_F$, this gives the
matching lower bound.  Hence,
\begin{align}
E_F(\rho_{AB}) = R + O(1).
\end{align}
Thus, in this parameter regime, the proposed relation between $E_F$ and
the entanglement wedge cross section is indeed satisfied at leading order.
To the best of our knowledge, this was not previously known.

The situation is different when
\begin{align}
0<A+B-C_A-C_B<R.
\end{align}
In this regime, the above argument does not determine $E_F$.
Indeed, our result gives
\begin{align}
E_R
=
A+B-C_A-C_B+O(1),
\end{align}
which can be parametrically smaller than $R$.
If the proposed relation $E_F=E_W+O(1)=R+O(1)$ continues to hold, this
would imply a leading-order separation between $E_R$ and $E_F$ already
in the two-tensor network.
We will address this question in a separate work.

\paragraph{Toward holographic theories:}

A direct formulation of the REE in continuum holographic theories involves an additional subtlety.
The usual notion of separability assumes a tensor-product decomposition into subsystems, whereas local operator algebras in continuum quantum field theory do not admit such a factorization in a straightforward manner~\cite{Harlow2017RT,Witten2022CrossedProduct}.
A satisfactory formulation may therefore require an algebraic definition of separability and of the relative entropy of entanglement itself~\cite{HollandsSanders2018}.
Instead, we worked with Haar random tensor networks, where the Hilbert spaces are finite dimensional and separability is unambiguous, while retaining the geometric structure that motivates the holographic proposal.

A possible gravity-side interpretation is to introduce a defect supported on the bulk tripartition surface $\Gamma$, with an action whose leading replica dependence is proportional to
\begin{align}
\frac{n-1}{4G_N}\,\Area(\Gamma).
\end{align}
At the level of the on-shell action, such a term would naturally reproduce the area contribution appearing in our result. 
However, unlike the usual cosmic-brane construction for entanglement entropy, it is not yet clear what underlying replica structure would give rise to the $n$-dependence above.
It would also be interesting to understand how such a defect backreacts on the bulk geometry, in particular when $\Gamma$ contains trijunctions.
Finally, the defect action alone does not explain why the corresponding boundary state should be separable.
A satisfactory continuum formulation should therefore clarify both the gravitational origin of the defect and the appropriate notion of separability in the boundary theory.

A complementary approach is to define the dephased state directly from
the boundary path integral.
An EPR pair corresponds, through the Choi--Jamio\l kowski isomorphism, to gluing two path integrals by an identity operator.
Dephasing additionally identifies the Schmidt labels between the ket and
bra, suggesting a modified gluing condition in which several sheets meet
along a common locus.
Such a construction may be related to the ``book'' boundary conditions
studied in many-body systems
~\cite{Aditya:2025yex,Ashida:2023ziz,Jiang:2025iet}.
It would be interesting to explore whether this perspective provides a new
way to formulate the dephased state directly in the boundary theory and to
study the corresponding relative entropy using CFT or BCFT techniques
~\cite{Caputa:2018xuf,TakayanagiUgajinUmemoto2018}.

\paragraph{Geometric interpretation of bond dephasing:}

It is also natural to ask whether the bond-dephased state appearing in our upper bound admits a direct bulk interpretation.
A useful analogy comes from projective measurements in tensor networks. Projecting the bonds along a surface onto fixed states replaces the contracted entangled bonds by product states, a construction reminiscent of inserting
end-of-world (EoW) brane-like objects along the measurement surface~\cite{Takayanagi2011BCFT,FujitaTakayanagiTonni2011, AntoniniEtAl2022, MoriYoshidaConnectedWedge}.
Bond dephasing can be viewed as performing the same measurement in the Schmidt basis and subsequently discarding the measurement outcome.
Thus, the dephased state corresponds not to a fixed postselected brane configuration, but rather to an incoherent mixture over the possible
measurement outcomes.
We will refer to this measure-and-forget object heuristically as a ``GHZ brane.''
It would be interesting to understand whether such an incoherent mixture of postselected configurations admits a useful semiclassical bulk description.
This may provide another route toward understanding the bulk meaning of the separable state used in our upper-bound construction.

\paragraph{Toward general random tensor networks:}

It is natural to ask whether the lower-bound argument can be generalized to
arbitrary random tensor networks.
While the structure of the proof based on sequential optimization gives some reason for optimism, there is an important technical complication.
Already in the two-tensor example, obtaining the optimal bound required
choosing the order of optimization appropriately.
A fixed order does not necessarily generate all candidate tripartition surfaces, and different optimization orders must be compared in order to recover the minimal one. 
Instead, it may be possible to design a suitable optimization order from the minimal surface. 
As the number of tensors increases, the number of possible optimization orders and intermediate subnetworks grows rapidly.
Developing a systematic way to organize these choices while retaining the globally minimal tripartition surface appears to be the main obstacle to extending the proof to general random tensor networks.

\paragraph{Implications for bulk reconstruction:}

The Ryu--Takayanagi surface has an important interpretation in bulk reconstruction where the entanglement wedge bounded by the minimal surface
identifies a bulk region whose operators can be reconstructed from the
corresponding boundary subsystem.
The appearance of a tripartition surface in the REE raises the possibility of an analogous multipartite reconstruction structure.
It would be interesting to understand what class of bulk degrees of freedom, if any, is naturally associated with the regions bounded by the tripartition surface and what boundary operations are required for their reconstruction.
We will explore this question separately.

\subsection*{Acknowledgment}

We thank 
Chitraang Murdia, Pratik Rath and Tadashi Takayanagi for useful discussions. 
Research at Perimeter Institute is supported in part by the Government of Canada through the Department of Innovation, Science and Economic Development and by the Province of Ontario through the Ministry of Colleges and Universities. 
This work was supported by the RIKEN TRIP initiative and JSPS KAKENHI Grant Number 23KJ1154, 24K17047.

\appendix

\section{Reduction to the Ising--OR model (Lemma~\ref{lemma:Ising})}
\label{app:binary-reduction}

In this appendix, we prove that the minimization of the replica action \eqref{eq:replica-action}
\begin{align}
I_{\Gamma}[\{g_v\}]
=
\sum_{e=(uv)\notin\Gamma}
\ell_e\, d(g_u,g_v)
+
\sum_{e=(uv)\in\Gamma}
\ell_e\, \delta(g_u,g_v)
\label{eq:app-replica-action}
\end{align}
over arbitrary replica spins $g_v\in S_q$ can be reduced to configurations with $g_v\in\{e,\tau\}$.

We first establish two useful inequalities.

\begin{lemma}
For arbitrary $g,h\in S_q$,
\begin{align}
d(g,h)
\geq
\left|
d(e,g)-d(e,h)
\right|.
\label{eq:app-cayley-triangle}
\end{align}
Furthermore, the dephased-edge distance satisfies
\begin{align}
\delta(g,h)
\geq
\max\left\{
d(e,g),d(e,h)
\right\}.
\label{eq:app-dephased-bound}
\end{align}
\end{lemma}

\begin{proof}
The first inequality follows immediately from the triangle inequality for
the Cayley distance.

To prove Eq.~\eqref{eq:app-dephased-bound}, recall that the equivalence relations determining $\kappa(g,h)$ include
\begin{align}
i_a=i_{g(a)},
\qquad
a=2,\ldots,q.
\label{eq:app-g-relations}
\end{align}
This relation leaves at most $\#(g)$ independent equivalence classes, as every cycle of $g$ gives one connected equivalence class.
Therefore
\begin{align}
\kappa(g,h)\leq \#(g).
\end{align}
It follows that
\begin{align}
\delta(g,h)
=
q-\kappa(g,h)
\geq
q-\#(g)
=
d(e,g)
\end{align}
as well as $\delta(g,h)\geq d(e,h)$ by exchanging $g$ and $h$. 
\end{proof}


We now associate a real-valued variable to each replica spin:
\begin{align}
x_v
:=
\frac{d(e,g_v)}{q-1}.
\label{eq:app-x-definition}
\end{align}
Since the maximal Cayley distance from the identity is $q-1$, we have
\begin{align}
0\leq x_v\leq 1.
\end{align}
The boundary conditions become
\begin{align}
x=0 \qquad &\text{on } C,
\nonumber\\
x=1 \qquad &\text{on } A,B.
\end{align}

Using Eqs.~\eqref{eq:app-cayley-triangle} and
\eqref{eq:app-dephased-bound}, the replica action obeys
\begin{align}
I_{\Gamma}[\{g_v\}]
\geq
(q-1)F_{\Gamma}[x],
\label{eq:app-action-lower-bound}
\end{align}
where
\begin{align}
F_{\Gamma}[x]
:=
\sum_{e=(uv)\notin\Gamma}
\ell_e\, |x_u-x_v|
+
\sum_{e=(uv)\in\Gamma}
\ell_e\, \max\{x_u,x_v\}.
\label{eq:app-continuous-action}
\end{align}

At this stage, $x_v$ is allowed to take arbitrary values in $[0,1]$.
We now use the standard thresholding representation of the Lovasz extension~\cite{Lovasz1983} to show that the minimum of Eq.~\eqref{eq:app-continuous-action} is always attained by a binary configuration $x_v=0,1$.

For each $t\in(0,1)$, define the thresholded variables
\begin{align}
s_v^{(t)}
:=
\begin{cases}
1, & x_v\geq t,\\
0, & x_v<t,
\end{cases}
\label{eq:app-threshold}
\end{align}
so we have a layer-cake representation of $x_v$ as
\begin{equation}
    x_v =\int_0^1 s_v^{(t)}\, dt.
    \label{eq:layer-cake-rep}
\end{equation}

Since $\abs{s_u^{(t)}-s_v^{(t)}}=1$ precisely when $t$ lies between
$x_u$ and $x_v$, we have
\begin{equation}
    |x_u-x_v|
    =
    \int_0^1
    \left|s_u^{(t)}-s_v^{(t)}\right|\,dt,
    \label{eq:app-layer-cake-ising}
\end{equation}
for any $x_u,x_v\in[0,1]$.

Similarly, $\max\{s_u^{(t)},s_v^{(t)}\}$ is nonzero, thus unity, only when $t\le x_u$ or $t\le x_v$. Since this is equivalent to $t\le \max\{x_u,x_v\}$, from Eq.~\eqref{eq:layer-cake-rep},
\begin{equation}
    \max\{x_u,x_v\} = \int_0^1 \max\left\{ s_u^{(t)},s_v^{(t)} \right\}\,dt
    \label{eq:app-layer-cake-or}
\end{equation}
for any $x_u,x_v\in[0,1]$.

Substituting these identities into Eq.~\eqref{eq:app-continuous-action},
we obtain
\begin{align}
F_{\Gamma}[x]
=
\int_0^1
F_{\Gamma}[s^{(t)}]\,dt.
\label{eq:app-coarea}
\end{align}
Because this is an averaged quantity, for any fractional configuration $x$, there exists a threshold
$t$ such that
\begin{align}
F_{\Gamma}[s^{(t)}]
\leq
F_{\Gamma}[x].
\end{align}
Thus allowing intermediate values $0<x<1$ cannot lower the minimum.
Since binary configurations are themselves allowed in the original minimization, the two minima are equal:
\begin{align}
\min_{\substack{
0\leq x_v\leq1\\
x|_A=x|_B=1,\,
x|_C=0}}
F_{\Gamma}[x]
=
\min_{\substack{
s_v\in\{0,1\}\\
s|_A=s|_B=1,\,
s|_C=0}}
F_{\Gamma}[s].
\label{eq:app-binary-min}
\end{align}

It remains to show that every binary configuration can be realized by replica permutations which saturate the bound Eq.~\eqref{eq:app-action-lower-bound}.
Given a binary configuration $\{s_v\}$, choose
\begin{align}
g_v
=
\begin{cases}
e, & s_v=0,\\
\tau, & s_v=1.
\end{cases}
\label{eq:app-binary-permutation}
\end{align}
For an undephased edge,
\begin{align}
d(g_u,g_v)
=
(q-1)|s_u-s_v|,
\end{align}
since
\begin{align}
d(e,e)=d(\tau,\tau)=0,
\qquad
d(e,\tau)=d(\tau,e)=q-1.
\end{align}
For a dephased edge,
\begin{align}
\delta(g_u,g_v)
=
(q-1)\max\{s_u,s_v\},
\end{align}
since
\begin{align}
\delta(e,e)=0,
\end{align}
while
\begin{align}
\delta(e,\tau)
=
\delta(\tau,e)
=
\delta(\tau,\tau)
=
q-1.
\end{align}
Hence the replica configuration in Eq.~\eqref{eq:app-binary-permutation} satisfies
\begin{align}
I_{\Gamma}[\{g_v\}]
=
(q-1)F_{\Gamma}[s].
\end{align}
Combining this observation with Eqs.~\eqref{eq:app-action-lower-bound} and
\eqref{eq:app-binary-min}, we conclude that
\begin{align}
\min_{\{g_v\in S_q\}}
I_{\Gamma}[\{g_v\}]
=
(q-1)
\min_{\substack{
s_v\in\{0,1\}\\
s|_A=s|_B=1,\,
s|_C=0}}
F_{\Gamma}[s]
\label{eq:app-final-binary-reduction}
\end{align}
which establishes Lemma~\ref{lemma:Ising}.

\section{On the $q=1$ limit}\label{app:q1-limit}

In the discussion above, we evaluated
$-\Tr[\rho_{AB}\log_2\sigma^\Gamma_{AB}]$ using the replica method,
by first computing the corresponding quantity at integer $q$ and then
formally taking the $q\to1$ limit.
Strictly speaking, one should verify that this limit is smooth and that
the leading large-$n$ behavior can be continued to $q=1$.

To address this point, define
\begin{align}
H_q^\Gamma
:=
-\frac{1}{q-1}
\log_2\Tr\!\left[
\rho_{AB}(\sigma^\Gamma_{AB})^{q-1}
\right].
\end{align}
This quantity is monotonically non-increasing with $q$, and therefore
\begin{align}
-\Tr\!\left[
\rho_{AB}\log_2\sigma^\Gamma_{AB}
\right]
=
H_1^\Gamma
\geq
H_2^\Gamma.
\end{align}
For $\Gamma=\Gamma_{\min}$, the $q=2$ replica calculation gives
\begin{align}
H_2^{\Gamma_{\min}}
=
\Area(\Gamma_{\min})+O(1).
\end{align}
Hence
\begin{align}
-\Tr\!\left[
\rho_{AB}\log_2\sigma^{\Gamma_{\min}}_{AB}
\right]
\geq
\Area(\Gamma_{\min})+O(1),
\end{align}
which establishes the desired leading-order lower bound at $q=1$.

A direct upper bound can be obtained without taking a replica limit.
For a fixed tripartition surface $\Gamma$, cut all bonds crossing
$\Gamma$ and collectively denote their Schmidt labels by $\mathbf i$.
Writing
\begin{align}
d_\Gamma:=\prod_{e\in\Gamma}D_e=2^{\Area(\Gamma)},
\end{align}
the normalized original and dephased states can be expressed as
\begin{align}
|\Psi\rangle
=
\frac{1}{\sqrt{N d_\Gamma}}
\sum_{\mathbf i}|\psi_{\mathbf i}\rangle,\qquad
\sigma^\Gamma_{ABC}
=
\frac{1}{N_\Gamma d_\Gamma}
\sum_{\mathbf i}
|\psi_{\mathbf i}\rangle\langle\psi_{\mathbf i}|,
\end{align}
where $N$ and $N_\Gamma$ are the normalization factors of the original
and dephased tensor networks, respectively.
Since there are $d_\Gamma$ configurations of $\mathbf i$, the
Cauchy--Schwarz inequality gives
\begin{align}
\rho_{ABC}
\leq
d_\Gamma\frac{N_\Gamma}{N}\,
\sigma^\Gamma_{ABC}
=
2^{\Area(\Gamma)}
\frac{N_\Gamma}{N}\,
\sigma^\Gamma_{ABC}.
\end{align}

Moreover, since $\sigma^\Gamma_{ABC}$ is separable across $A:B:C$, we may
write
\begin{align}
\sigma^\Gamma_{ABC}
=
\sum_i p_i\,
\sigma_A^{(i)}\otimes\sigma_B^{(i)}\otimes\sigma_C^{(i)}.
\end{align}
Since each $\sigma_C^{(i)}$ is a density matrix and therefore satisfies
$\sigma_C^{(i)}\leq I_C$, it follows that
\begin{align}
\sigma^\Gamma_{ABC}
\leq
\sigma^\Gamma_{AB}\otimes I_C .
\end{align}
Hence
\begin{align}
\rho_{ABC}
\leq
2^{\Area(\Gamma)}
\frac{N_\Gamma}{N}
\left(\sigma^\Gamma_{AB}\otimes I_C\right).
\end{align}
Since $\rho_{ABC}=|\Psi\rangle\langle\Psi|$ is pure, this implies
\begin{align}
\left\langle\Psi\left|
\left((\sigma^\Gamma_{AB})^{-1}\otimes I_C\right)
\right|\Psi\right\rangle
\leq
2^{\Area(\Gamma)}
\frac{N_\Gamma}{N},
\end{align}
where the inverse is understood on the support of
$\sigma^\Gamma_{AB}$.
Using Jensen's inequality, we therefore obtain
\begin{align}
-\Tr\!\left[
\rho_{AB}\log_2\sigma^\Gamma_{AB}
\right]
&=
\left\langle\Psi\left|
\log_2\!\left((\sigma^\Gamma_{AB})^{-1}\otimes I_C\right)
\right|\Psi\right\rangle
\nonumber\\
&\leq
\Area(\Gamma)
+
\log_2\frac{N_\Gamma}{N}.
\end{align}
The normalization fluctuations discussed above give
$\log_2(N_\Gamma/N)=o(1)$, and hence
\begin{align}
-\Tr\!\left[
\rho_{AB}\log_2\sigma^\Gamma_{AB}
\right]
\leq
\Area(\Gamma)+o(1).
\end{align}
This established the matching upper and lower bounds for $\Gamma = \Gamma_{\mathrm{min}}$.

\section{Proofs of product-state overlap bounds}
\label{app:product-overlap-proofs}

In this appendix, we provide the proofs of the two results used
in Sec.~\ref{sec:bound}.
We first derive the general lower bound on the relative entropy
of entanglement in terms of the maximal product-state overlap.
We then evaluate the overlap with product states constructed
from bulk tripartition surfaces.

\subsection{Lower bound on the relative entropy of entanglement}

We begin by proving Eq.~\eqref{eq:REE-product-overlap-bound}.
Let $\ket{\Psi}_{ABC}$ be a tripartite pure state and
$\rho_{AB}=\Tr_C\ket{\Psi}\bra{\Psi}$.

First, note that the maximal tripartite product overlap can
equivalently be expressed in terms of $\rho_{AB}$:
\begin{align}
\Lambda(\Psi_{A:B:C})
=
\max_{\alpha,\beta}
\langle\alpha,\beta|
\rho_{AB}
|\alpha,\beta\rangle.
\end{align}
To see this, for fixed $\alpha$ and $\beta$, define
\begin{align}
\ket{v_{\alpha,\beta}}_C
:=
{}_{AB}\langle\alpha,\beta|\Psi\rangle_{ABC}.
\end{align}
Then
\begin{align}
\langle\alpha,\beta|
\rho_{AB}
|\alpha,\beta\rangle
&=
\|v_{\alpha,\beta}\|^2
\nonumber\\
&=
\max_\gamma
|\langle\alpha,\beta,\gamma|\Psi\rangle|^2.
\end{align}
Maximizing over $\alpha$ and $\beta$ establishes
Eq.~\eqref{eq:pure-mixed-product-overlap}.

Now let $\sigma_{AB}$ be an arbitrary separable state.
If
$\operatorname{supp}\rho_{AB}
\not\subseteq
\operatorname{supp}\sigma_{AB}$,
then $D(\rho_{AB}\Vert\sigma_{AB})=+\infty$,
and the desired inequality is immediate.
We may therefore assume
$\operatorname{supp}\rho_{AB}
\subseteq
\operatorname{supp}\sigma_{AB}$.

By concavity of the logarithm,
\begin{align}
\Tr[\rho_{AB}\log_2\sigma_{AB}]
\leq
\log_2\Tr[\rho_{AB}\sigma_{AB}].
\label{eq:log-concavity}
\end{align}
Since $\sigma_{AB}$ is separable, it can be written as
\begin{align}
\sigma_{AB}
=
\sum_i p_i
\ket{a_i,b_i}\bra{a_i,b_i},
\end{align}
where $p_i\geq 0$ and $\sum_i p_i=1$.
Consequently,
\begin{align}
\Tr[\rho_{AB}\sigma_{AB}]
&=
\sum_i p_i
\langle a_i,b_i|
\rho_{AB}
|a_i,b_i\rangle
\nonumber\\
&\leq
\Lambda(\rho_{A:B})
=
\Lambda(\Psi_{A:B:C}).
\end{align}
Combining these inequalities gives
\begin{align}
D(\rho_{AB}\Vert\sigma_{AB})
&=
-S(\rho_{AB})
-
\Tr[\rho_{AB}\log_2\sigma_{AB}]
\nonumber\\
&\geq
-S(\rho_{AB})
-
\log_2\Lambda(\Psi_{A:B:C}).
\end{align}
Since this holds for every separable state $\sigma_{AB}$,
minimizing over $\sigma_{AB}\in\mathrm{SEP}(A:B)$ yields
\begin{align}
E_R(\rho_{A:B})
\geq
-S(\rho_{AB})
-
\log_2\Lambda(\Psi_{A:B:C}),
\end{align}
as claimed.

\subsection{Product-state overlap from tripartition surfaces}

We now prove Eqs.~\eqref{eq:product-overlap-minimization}
and \eqref{eq:product-overlap-minimal-surface}.

For a given bulk tripartition surface $\Gamma$,
let $\ket{\Psi^{\Gamma,|0\rangle}_{ABC}}$
be the normalized pure state obtained by replacing
every EPR bond crossing $\Gamma$ with a fixed product
state $\ket{0}_{\Gamma}$.
Since fixing these bonds disconnects the three bulk regions,
the resulting boundary state factorizes across $A:B:C$,
as in Eq.~\eqref{eq:product-state-factorization}.

Define
\begin{align}
\rho_{ABC}^{\Gamma,|0\rangle}
:=
\ket{\Psi^{\Gamma,|0\rangle}_{ABC}}
\bra{\Psi^{\Gamma,|0\rangle}_{ABC}}.
\end{align}
We evaluate the squared overlap
\begin{align}
X_2^\Gamma
&:=
\left|
\langle\Psi^{\Gamma,|0\rangle}_{ABC}
|\Psi_{ABC}\rangle
\right|^2
\nonumber\\
&=
\Tr\left[
\rho_{ABC}
\rho_{ABC}^{\Gamma,|0\rangle}
\right].
\label{eq:product-state-overlap}
\end{align}

This quantity can be evaluated using a two-replica calculation.
Replica 1 comes from $\rho_{ABC}$, while replica 2 comes
from $\rho_{ABC}^{\Gamma,|0\rangle}$.
Since the trace is taken over the full boundary $ABC$,
the boundary conditions are
\begin{align}
g=\tau
\qquad
\text{on } A,B,C,
\label{eq:product-overlap-boundary-condition}
\end{align}
where $\tau\in S_2$ is the swap permutation.

Introducing the Ising-valued variable
\begin{align}
s_v
=
\begin{cases}
0, & g_v=e,\\
1, & g_v=\tau,
\end{cases}
\end{align}
the ordinary edge interaction is
\begin{align}
-\log_2 W_e(g_u,g_v)
=
\ell_e|s_u-s_v|,
\qquad e\notin\Gamma.
\end{align}
For an edge crossing $\Gamma$, the fixed product-state
projection gives
\begin{align}
-\log_2 W_e^{(0)}(g_u,g_v)
=
\ell_e\max\{s_u,s_v\},
\qquad e\in\Gamma.
\end{align}
Thus, the two-replica calculation leads to the same
Ising--OR interaction as in Sec.~2,
but with the boundary condition $s=1$ on all three
regions $A$, $B$, and $C$.

Define
\begin{align}
T_{111}(\Gamma)
:=
\min_{\substack{
s_v\in\{0,1\}\\
s|_A=s|_B=s|_C=1
}}
\left[
\sum_{e=(uv)\notin\Gamma}
\ell_e|s_u-s_v|
+
\sum_{e=(uv)\in\Gamma}
\ell_e\max\{s_u,s_v\}
\right].
\label{eq:T111-definition}
\end{align}
The leading-order overlap is then given by
\begin{align}
-\log_2 X_2^\Gamma
=
T_{111}(\Gamma)+O(1).
\label{eq:product-overlap-result}
\end{align}

We now show that
\begin{align}
\mathrm{Area}(\Gamma_{\min})
\leq
T_{111}(\Gamma)
\leq
\mathrm{Area}(\Gamma).
\label{eq:product-overlap-bounds}
\end{align}

The upper bound follows immediately by choosing
the uniform configuration $s_v=1$ at every vertex.
All ordinary edges then have zero energy,
while every edge of $\Gamma$ contributes its full weight.
Hence,
\begin{align}
T_{111}(\Gamma)
\leq
\mathrm{Area}(\Gamma).
\end{align}

For the lower bound, we use the same energy-carrying-edge
argument developed in Sec.~2.3.
For a given Ising configuration, consider the set of
ordinary edges with $s_u\neq s_v$, together with
the edges of $\Gamma$ for which $\max\{s_u,s_v\}=1$.

We claim that these energy-carrying edges separate
the three boundary regions $A$, $B$, and $C$.
Indeed, consider a path connecting any two of these
boundary regions.
If the path avoids all energy-carrying edges,
its spin must remain $s=1$ throughout, since all
three boundary regions have $s=1$.
Moreover, such a path cannot cross $\Gamma$,
because every edge of $\Gamma$ adjacent to an
$s=1$ vertex contributes nonzero OR energy.
This contradicts the fact that $\Gamma$ separates
$A$, $B$, and $C$.

Thus, the energy-carrying edges contain a valid
tripartition surface.
Since all edge weights are positive, their total
weight is at least $\mathrm{Area}(\Gamma_{\min})$.
Consequently,
\begin{align}
T_{111}(\Gamma)
\geq
\mathrm{Area}(\Gamma_{\min}).
\end{align}
Together with the upper bound, this establishes
Eq.~\eqref{eq:product-overlap-bounds}.

In particular, choosing $\Gamma=\Gamma_{\min}$
gives
\begin{align}
T_{111}(\Gamma_{\min})
=
\mathrm{Area}(\Gamma_{\min}).
\end{align}
Using Eq.~\eqref{eq:product-overlap-result}, we obtain
\begin{align}
-\log_2
\left|
\langle
\Psi^{\Gamma_{\min},|0\rangle}_{ABC}
|
\Psi_{ABC}
\rangle
\right|^2
=
\mathrm{Area}(\Gamma_{\min})+O(1).
\end{align}

Furthermore, for any tripartition surface $\Gamma$,
Eq.~\eqref{eq:product-overlap-bounds} implies
\begin{align}
-\log_2
\left|
\langle
\Psi^{\Gamma,|0\rangle}_{ABC}
|
\Psi_{ABC}
\rangle
\right|^2
\geq
\mathrm{Area}(\Gamma_{\min})+O(1).
\end{align}
Therefore, the product state constructed from
$\Gamma_{\min}$ achieves the maximal overlap
at leading order within this geometric family,
establishing
Eqs.~\eqref{eq:product-overlap-minimization}
and \eqref{eq:product-overlap-minimal-surface}.

\section{Maximal product overlap for one Haar random tensor}
\label{app:one-tensor-overlap}

In this appendix, we prove Lemma~\ref{lem:overlap}.
We first recall the Gaussian concentration inequality and
introduce the projected tensor notation.
We then estimate the maximal product overlap using
$\epsilon$-net arguments.

\subsection{Gaussian concentration and normalization}

The Gaussian tensor $|G\rangle$ defined in
Sec.~\ref{sec:one-random-tensor} satisfies
\begin{align}
\mathbb E\langle G|G\rangle=1.
\end{align}
Its norm is sharply concentrated around unity,
\begin{align}
\Pr\left[
\left|
\langle G|G\rangle-1
\right|>t
\right]
\leq
K e^{-\eta t^2abc}
\label{eq:one-tensor-Gaussian-norm-concentration}
\end{align}
for $0<t\leq 1$, with universal constants $K,\eta>0$.
Therefore, the normalization factor
in Eq.~\eqref{eq:one-tensor-Haar-state} does not
affect the leading-order estimates below.

We will use the standard Gaussian concentration
inequality~\cite{Ledoux2001}.

\begin{lemma}\label{lem:Gaussian-concentration}
If $g$ is a standard complex Gaussian vector and $f(g)$ is
$L$-Lipschitz, then
\begin{align}
\Pr\left[
f(g)>
\mathbb Ef(g)+tL
\right]
\leq
e^{-\eta t^2}
\label{eq:Gaussian-concentration-Lipschitz}
\end{align}
with some universal constant $\eta>0$.
\end{lemma}

Here, the Lipschitz constant $L$ is defined by
\begin{align}
|f(g)-f(g')|
\leq
L\|g-g'\|
\end{align}
for arbitrary Gaussian vectors $g$ and $g'$.

\subsection{Overlap for fixed product states}

For normalized states
$|\alpha\rangle\in\mathcal H_A$ and
$|\beta\rangle\in\mathcal H_B$, define
\begin{align}
|G_\alpha\rangle_{BC}
&:=
{}_A\langle\alpha|G\rangle_{ABC}
=
\figbox{2.0}{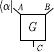},
\label{eq:one-tensor-G-alpha}
\\
|G_{\alpha,\beta}\rangle_C
&:=
{}_{AB}\langle\alpha,\beta|G\rangle_{ABC}
=
\figbox{2.0}{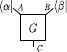},
\label{eq:one-tensor-G-alpha-beta}
\\
G_{\alpha,\beta,\gamma}
&:=
{}_{ABC}\langle\alpha,\beta,\gamma|G\rangle_{ABC}
=
\figbox{2.0}{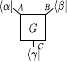}.
\label{eq:one-tensor-G-alpha-beta-gamma}
\end{align}
Here, contracted tensor indices are indicated
by subscripts of the tensor $G$.

Up to the normalization factor discussed above,
the maximal product overlap is
\begin{align}
\Lambda(\rho_{A:B})
=
\max_{\alpha,\beta,\gamma}
|G_{\alpha,\beta,\gamma}|^2
=
\max_{\alpha,\beta}
\bigl\|\,|G_{\alpha,\beta}\rangle\,\bigr\|^2.
\label{eq:one-tensor-overlap-projected}
\end{align}

We first estimate
$\bigl\|\,|G_{\alpha,\beta}\rangle\,\bigr\|$
for fixed $\alpha$ and $\beta$.
Using the Gaussian covariance,
\begin{align}
\mathbb E
\bigl\|\,|G_{\alpha,\beta}\rangle\,\bigr\|^2
=
\frac{1}{ab},
\qquad
\mathbb E
\bigl\|\,|G_{\alpha,\beta}\rangle\,\bigr\|
\leq
\frac{1}{\sqrt{ab}}.
\label{eq:one-tensor-fixed-product-mean}
\end{align}

To apply Gaussian concentration, introduce the
rescaled projected vector
\begin{align}
|g\rangle_C
:=
\sqrt{abc}\,
|G_{\alpha,\beta}\rangle_C,
\end{align}
which is a standard complex Gaussian vector
on $\mathcal H_C$.
Define
\begin{align}
f_{\alpha,\beta}(g)
:=
\frac{1}{\sqrt{abc}}\,
\bigl\|\,|g\rangle\,\bigr\|.
\end{align}
By construction,
$f_{\alpha,\beta}(g)
=
\bigl\|\,|G_{\alpha,\beta}\rangle\,\bigr\|$.

For two arbitrary vectors $g,g'$, the reverse
triangle inequality gives
\begin{align}
|f_{\alpha,\beta}(g)-f_{\alpha,\beta}(g')|
&=
\frac{1}{\sqrt{abc}}
\left|
\|g\|-\|g'\|
\right|
\nonumber\\
&\leq
\frac{1}{\sqrt{abc}}
\|g-g'\|.
\label{eq:one-tensor-fixed-product-Lipschitz}
\end{align}
Thus, $f_{\alpha,\beta}$ is $L$-Lipschitz with
\begin{align}
L=\frac{1}{\sqrt{abc}}.
\end{align}

The Gaussian concentration inequality gives
\begin{align}
\Pr\left[
\bigl\|\,|G_{\alpha,\beta}\rangle\,\bigr\|
>
\frac{1}{\sqrt{ab}}
+
\frac{t}{\sqrt{abc}}
\right]
\leq
e^{-\eta t^2}.
\label{eq:one-tensor-fixed-product-tail}
\end{align}
This establishes the concentration estimate
for fixed $\alpha$ and $\beta$.

\subsection{Optimization over product states}

We now optimize over $\alpha$ and $\beta$
using $\epsilon$-nets.

First, fix $\alpha$ and let $\mathcal N_B$
be an $\epsilon$-net of the unit sphere in
$\mathbb C^b$, with $0<\epsilon<1/2$.
The net can be chosen such that~\cite{HaydenLeungShorWinter2004}
\begin{align}
|\mathcal N_B|
\leq
\left(
1+\frac{2}{\epsilon}
\right)^{2b}
\leq
e^{Kb}.
\label{eq:one-tensor-B-net-size}
\end{align}

Applying Eq.~\eqref{eq:one-tensor-fixed-product-tail}
to every $\beta_0\in\mathcal N_B$ and using
the union bound gives
\begin{align}
\Pr\left[
\max_{\beta_0\in\mathcal N_B}
\bigl\|\,|G_{\alpha,\beta_0}\rangle\,\bigr\|
>
\frac{1}{\sqrt{ab}}
+
\frac{t}{\sqrt{abc}}
\right]
\leq
e^{Kb-\eta t^2}.
\label{eq:one-tensor-B-net-union-bound}
\end{align}

Shifting the deviation parameter by a term of
order $\sqrt b$ and adjusting the universal
constants, we obtain
\begin{align}
\Pr\left[
\max_{\beta_0\in\mathcal N_B}
\bigl\|\,|G_{\alpha,\beta_0}\rangle\,\bigr\|
>
K
\left(
\frac{1}{\sqrt{ab}}
+
\frac{1}{\sqrt{ac}}
+
\frac{t}{\sqrt{abc}}
\right)
\right]
\leq
e^{-\eta t^2}.
\label{eq:one-tensor-B-net-bound}
\end{align}

We next extend this estimate from the
$\epsilon$-net to arbitrary normalized
$|\beta\rangle\in\mathcal H_B$.

For every normalized $|\beta\rangle$, there exists
$|\beta_0\rangle\in\mathcal N_B$ such that
$\|\beta-\beta_0\|\leq\epsilon$.
By the triangle inequality and linearity
of the projected tensor,
\begin{align}
\bigl\|\,|G_{\alpha,\beta}\rangle\,\bigr\|
&\leq
\bigl\|\,|G_{\alpha,\beta_0}\rangle\,\bigr\|
+
\bigl\|\,|G_{\alpha,\beta-\beta_0}\rangle\,\bigr\|
\nonumber\\
&\leq
\bigl\|\,|G_{\alpha,\beta_0}\rangle\,\bigr\|
+
\epsilon
\max_{\|\beta\|=1}
\bigl\|\,|G_{\alpha,\beta}\rangle\,\bigr\|.
\end{align}
Taking the maximum over $\beta$ gives
\begin{align}
\max_{\|\beta\|=1}
\bigl\|\,|G_{\alpha,\beta}\rangle\,\bigr\|
\leq
\frac{1}{1-\epsilon}
\max_{\beta_0\in\mathcal N_B}
\bigl\|\,|G_{\alpha,\beta_0}\rangle\,\bigr\|.
\label{eq:one-tensor-B-net-to-sphere}
\end{align}

Since $\epsilon$ is fixed, the factor
$(1-\epsilon)^{-1}$ can be absorbed into $K$.
Hence,
\begin{align}
\Pr\left[
\max_{\beta}
\bigl\|\,|G_{\alpha,\beta}\rangle\,\bigr\|
>
K
\left(
\frac{1}{\sqrt{ab}}
+
\frac{1}{\sqrt{ac}}
+
\frac{t}{\sqrt{abc}}
\right)
\right]
\leq
e^{-\eta t^2}.
\label{eq:one-tensor-fixed-alpha-bound}
\end{align}

We finally optimize over $\alpha$.
Let $\mathcal N_A$ be an $\epsilon$-net
of the unit sphere in $\mathbb C^a$, satisfying
\begin{align}
|\mathcal N_A|\leq e^{Ka}.
\end{align}
Applying Eq.~\eqref{eq:one-tensor-fixed-alpha-bound}
to every $\alpha_0\in\mathcal N_A$,
using the union bound, and shifting the
deviation parameter by a term of order
$\sqrt a$, we obtain
\begin{align}
\Pr\left[
\max_{\alpha_0\in\mathcal N_A}
\max_\beta
\bigl\|\,|G_{\alpha_0,\beta}\rangle\,\bigr\|
>
K
\left(
\frac{1}{\sqrt{ab}}
+
\frac{1}{\sqrt{ac}}
+
\frac{1}{\sqrt{bc}}
+
\frac{t}{\sqrt{abc}}
\right)
\right]
\leq
e^{-\eta t^2}.
\label{eq:one-tensor-A-net-bound}
\end{align}

Applying the same net-to-sphere argument
to $\alpha$ yields
\begin{align}
\Pr\left[
\max_{\alpha,\beta}
\bigl\|\,|G_{\alpha,\beta}\rangle\,\bigr\|
>
K
\left(
\frac{1}{\sqrt{ab}}
+
\frac{1}{\sqrt{ac}}
+
\frac{1}{\sqrt{bc}}
+
\frac{t}{\sqrt{abc}}
\right)
\right]
\leq
e^{-\eta t^2}.
\label{eq:one-tensor-full-tail}
\end{align}

Finally, choosing
\begin{align}
t=\sqrt a+\sqrt b+\sqrt c
\end{align}
and adjusting the universal constants, we obtain
\begin{align}
\max_{\alpha,\beta}
\bigl\|\,|G_{\alpha,\beta}\rangle\,\bigr\|^2
\leq
K
\left(
\frac{1}{ab}
+
\frac{1}{ac}
+
\frac{1}{bc}
\right)
\end{align}
except with probability at most
$e^{-\eta(a+b+c)}$.

The normalization concentration
Eq.~\eqref{eq:one-tensor-Gaussian-norm-concentration}
allows us to replace the Gaussian vector $|G\rangle$
by the normalized Haar state $|\Psi\rangle$,
up to a change in the universal constants.
Therefore,
\begin{align}
\Lambda(\rho_{A:B})
\leq
K
\left(
\frac{1}{ab}
+
\frac{1}{ac}
+
\frac{1}{bc}
\right)
\end{align}
with the claimed failure probability.
This proves Lemma~\ref{lem:overlap}.

\section{Maximal product overlap for two Haar random tensors}
\label{app:two-tensor-overlap}

In this appendix, we prove Lemma~\ref{lemma:2TN}.
The proof extends the Gaussian concentration and $\epsilon$-net
arguments of Appendix~\ref{app:one-tensor-overlap} by
applying them sequentially to the two independent random tensors.
We first derive a concentration bound for $H$ while keeping
$G$ fixed, and then optimize over the product states
$\alpha$ and $\beta$.
Finally, we apply the one-tensor results to $G$ and
compare the two possible optimization orders.

\subsection{Conditional Gaussian concentration}

We begin by estimating
\begin{align}
\sqrt r\,
\bigl\|\,|G_\alpha\star H_\beta\rangle\,\bigr\|
\end{align}
for fixed normalized states $\alpha$ and $\beta$.
We first consider the probability only over the Gaussian
tensor $H$, while keeping $G$ fixed.
Importantly, no randomness of $G$ is used at this stage,
so the following argument applies to an arbitrary fixed tensor $G$.

Recall the projected tensors
\begin{align}
|G_\alpha\rangle_{C_A R}
&:=
{}_A\langle\alpha|G\rangle_{A C_A R},
\\
|H_\beta\rangle_{C_B R}
&:=
{}_B\langle\beta|H\rangle_{B C_B R}.
\end{align}
The contraction over the common $R$ leg gives
$|G_\alpha\star H_\beta\rangle_{C_A C_B}$.

Averaging over $H$, we obtain
\begin{align}
\mathbb E_H\left[
r\,
|G_\alpha\star H_\beta\rangle
\langle G_\alpha\star H_\beta|
\right]
=
\frac{1}{b}\,
\Tr_R\left[
|G_\alpha\rangle\langle G_\alpha|
\right]
\otimes
\frac{I_{C_B}}{c_B}.
\label{eq:two-tensor-fixed-alpha-beta-second-moment}
\end{align}
Taking the trace gives
\begin{align}
\mathbb E_H\left[
r\,
\bigl\|\,|G_\alpha\star H_\beta\rangle\,\bigr\|^2
\right]
=
\frac{1}{b}
\bigl\|\,|G_\alpha\rangle\,\bigr\|^2.
\end{align}
By Jensen's inequality,
\begin{align}
\mathbb E_H\left[
\sqrt r\,
\bigl\|\,|G_\alpha\star H_\beta\rangle\,\bigr\|
\right]
\leq
\frac{1}{\sqrt b}
\bigl\|\,|G_\alpha\rangle\,\bigr\|.
\label{eq:two-tensor-conditional-mean}
\end{align}
Since we have averaged only over $H$, the remaining
norm depends on $G$ and $\alpha$.

The quantity $\|G_\alpha\|$ has a simple interpretation
as a bipartite product amplitude,
\begin{align}
\bigl\|\,|G_\alpha\rangle\,\bigr\|
=
\max_{
|\phi\rangle\in\mathcal H_{C_A}\otimes\mathcal H_R
}
|\langle\alpha,\phi|G\rangle|.
\end{align}
Here, $\phi$ is an arbitrary state on the combined
subsystem $C_A R$ and need not factorize across $C_A:R$.

To apply Gaussian concentration, introduce the
rescaled projected vector
\begin{align}
|h\rangle_{C_B R}
:=
\sqrt{b c_B r}\,
|H_\beta\rangle_{C_B R}.
\label{eq:two-tensor-standard-Gaussian-vector}
\end{align}
By the covariance of $H$, $|h\rangle$ is a
standard complex Gaussian vector on
$\mathcal H_{C_B}\otimes\mathcal H_R$.
We define
\begin{align}
f_{\alpha,\beta;G}(h)
:=
\frac{1}{\sqrt{b c_B}}\,
\bigl\|\,|G_\alpha\star h\rangle\,\bigr\|.
\label{eq:two-tensor-fixed-product-function}
\end{align}
By construction,
\begin{align}
f_{\alpha,\beta;G}(h)
=
\sqrt r\,
\bigl\|\,|G_\alpha\star H_\beta\rangle\,\bigr\|.
\end{align}

We next estimate the Lipschitz constant of
$f_{\alpha,\beta;G}$.
For this purpose, define
\begin{align}
\lambda_\alpha(G)
:=
\max_{
|\gamma\rangle\in\mathcal H_{C_A},\,
|u\rangle\in\mathcal H_R
}
|\langle\gamma,u|G_\alpha\rangle|.
\label{eq:two-tensor-lambda-alpha}
\end{align}
This is the maximal product amplitude of
$|G_\alpha\rangle$ across $C_A:R$.

Consider two arbitrary vectors $h,h'$ on $C_B R$
and let
\begin{align}
|\delta h\rangle
:=
|h-h'\rangle.
\end{align}
Using an orthonormal basis $|\nu\rangle$ of $C_B$,
we can decompose
\begin{align}
|\delta h\rangle_{C_B R}
=
\sum_\nu
|\nu\rangle_{C_B}
|\delta h_\nu\rangle_R.
\label{eq:two-tensor-delta-h-decomposition}
\end{align}
The variation of the contracted tensor is
\begin{align}
|G_\alpha\star\delta h\rangle_{C_A C_B}
=
\sum_\nu
|G_{\alpha,\delta h_\nu}\rangle_{C_A}
\otimes
|\nu\rangle_{C_B},
\label{eq:two-tensor-varied-contraction}
\end{align}
where
\begin{align}
|G_{\alpha,\delta h_\nu}\rangle_{C_A}
:=
{}_{AR}\langle\alpha,\delta h_\nu|G\rangle_{A C_A R}.
\end{align}

Since the states $|\nu\rangle$ are orthogonal,
\begin{align}
\bigl\|\,|G_\alpha\star\delta h\rangle\,\bigr\|^2
=
\sum_\nu
\bigl\|\,|G_{\alpha,\delta h_\nu}\rangle\,\bigr\|^2.
\end{align}
By the definition of $\lambda_\alpha(G)$,
\begin{align}
\bigl\|\,|G_{\alpha,\delta h_\nu}\rangle\,\bigr\|
\leq
\lambda_\alpha(G)
\bigl\|\,|\delta h_\nu\rangle\,\bigr\|.
\end{align}
Consequently,
\begin{align}
\bigl\|\,|G_\alpha\star\delta h\rangle\,\bigr\|^2
&\leq
\lambda_\alpha(G)^2
\sum_\nu
\bigl\|\,|\delta h_\nu\rangle\,\bigr\|^2
\nonumber\\
&=
\lambda_\alpha(G)^2
\bigl\|\,|\delta h\rangle\,\bigr\|^2.
\label{eq:two-tensor-varied-contraction-bound}
\end{align}

Using the reverse triangle inequality, we therefore obtain
\begin{align}
\left|
f_{\alpha,\beta;G}(h)
-
f_{\alpha,\beta;G}(h')
\right|
&\leq
\frac{1}{\sqrt{b c_B}}
\bigl\|\,|G_\alpha\star\delta h\rangle\,\bigr\|
\nonumber\\
&\leq
\frac{\lambda_\alpha(G)}{\sqrt{b c_B}}\,
\bigl\|\,|\delta h\rangle\,\bigr\|.
\label{eq:two-tensor-fixed-alpha-beta-Lipschitz}
\end{align}
Thus, $f_{\alpha,\beta;G}$ is $L_\alpha$-Lipschitz with
\begin{align}
L_\alpha
=
\frac{\lambda_\alpha(G)}{\sqrt{b c_B}}.
\label{eq:two-tensor-fixed-alpha-beta-Lipschitz-constant}
\end{align}

Combining this with the conditional expectation
Eq.~\eqref{eq:two-tensor-conditional-mean},
the Gaussian concentration inequality gives
\begin{align}
\Pr_H\left[
\sqrt r\,
\bigl\|\,|G_\alpha\star H_\beta\rangle\,\bigr\|
>
\frac{
\bigl\|\,|G_\alpha\rangle\,\bigr\|
}{\sqrt b}
+
\frac{t}{\sqrt{b c_B}}\,
\lambda_\alpha(G)
\right]
\leq
e^{-\eta t^2}.
\label{eq:two-tensor-fixed-alpha-beta-tail}
\end{align}

The crucial point is that this concentration
bound holds for every fixed $G$.
The expectation and fluctuations are controlled
by the two quantities $\|G_\alpha\|$
and $\lambda_\alpha(G)$, respectively,
which involve only the remaining tensor $G$.

\subsection{Sequential optimization over product states}

We now optimize over $\alpha$ and $\beta$.
The argument follows the $\epsilon$-net method
developed in Appendix~\ref{app:one-tensor-overlap},
but the probability estimates remain conditional
on the tensor $G$.

First, fix $\alpha$ and $G$.
Let $\mathcal N_B$ be an $\epsilon$-net of
the unit sphere in $\mathbb C^b$, with
$0<\epsilon<1/2$, satisfying
\begin{align}
|\mathcal N_B|
\leq
e^{Kb}.
\end{align}
Applying Eq.~\eqref{eq:two-tensor-fixed-alpha-beta-tail}
to every $\beta_0\in\mathcal N_B$ and using
the union bound gives
\begin{align}
\Pr_H\left[
\exists\beta_0\in\mathcal N_B:
\sqrt r\,
\bigl\|\,|G_\alpha\star H_{\beta_0}\rangle\,\bigr\|
>
\frac{\|G_\alpha\|}{\sqrt b}
+
\frac{t}{\sqrt{b c_B}}
\lambda_\alpha(G)
\right]
\leq
e^{Kb-\eta t^2}.
\label{eq:two-tensor-B-net-union-bound}
\end{align}

Shifting the deviation parameter by a term
of order $\sqrt b$ compensates for the
number of representative states.
Using the standard net-to-sphere argument
and absorbing fixed factors into $K$, we obtain
\begin{align}
\Pr_H\left[
\max_\beta
\sqrt r\,
\bigl\|\,|G_\alpha\star H_\beta\rangle\,\bigr\|
>
K\left(
\frac{\|G_\alpha\|}{\sqrt b}
+
\left[
\frac{1}{\sqrt{c_B}}
+
\frac{t}{\sqrt{b c_B}}
\right]
\lambda_\alpha(G)
\right)
\right]
\leq
e^{-\eta t^2}.
\label{eq:two-tensor-fixed-alpha-optimized-beta}
\end{align}

We next optimize over $\alpha$.
Let $\mathcal N_A$ be an $\epsilon$-net of
the unit sphere in $\mathbb C^a$, with
\begin{align}
|\mathcal N_A|\leq e^{Ka}.
\end{align}
Applying Eq.~\eqref{eq:two-tensor-fixed-alpha-optimized-beta}
to every $\alpha_0\in\mathcal N_A$,
using the union bound, and shifting
the deviation parameter by a term of order
$\sqrt a$, we obtain the corresponding
bound for all representative states $\alpha_0$.

The net-to-sphere argument can then be applied
to $\alpha$ as in the one-tensor case.
Since the quantities on the right-hand side
depend on $\alpha$, we bound them by their
maxima over all normalized $\alpha$.
This gives
\begin{align}
\Pr_H\left[
\sqrt r\,
\max_{\alpha,\beta}
\bigl\|\,|G_\alpha\star H_\beta\rangle\,\bigr\|
>
K\left(
\frac{1}{\sqrt b}
\max_\alpha\|G_\alpha\|
+
\frac{\sqrt a+\sqrt b+t}{\sqrt{b c_B}}
\max_\alpha\lambda_\alpha(G)
\right)
\right]
\leq
e^{-\eta t^2}.
\label{eq:two-tensor-alpha-beta-optimized}
\end{align}

Importantly, this remains a conditional probability
bound over $H$ for every fixed $G$.
We can therefore first control the optimization
over $\alpha$ and $\beta$ using the randomness of $H$,
and only afterward apply probabilistic bounds
to the remaining tensor $G$.

We now evaluate the two quantities involving $G$.

The first is the maximal bipartite product amplitude
across $A:(C_A R)$,
\begin{align}
\max_\alpha\|G_\alpha\|
=
\max_{\alpha,\phi}
|\langle\alpha,\phi|G\rangle|.
\end{align}
The Gaussian concentration and $\epsilon$-net
argument for a bipartite tensor gives
\begin{align}
\Pr_G\left[
\max_\alpha\|G_\alpha\|
>
K\left(
\frac{1}{\sqrt a}
+
\frac{1}{\sqrt{c_A r}}
\right)
\right]
\leq
e^{-\eta(a+c_A r)}.
\label{eq:two-tensor-G-bipartite-overlap}
\end{align}
The local triangle inequality $A<C_A+R$
implies $a<c_A r$, so
\begin{align}
\max_\alpha\|G_\alpha\|
\leq
\frac{K}{\sqrt a}
\label{eq:two-tensor-G-bipartite-simplified}
\end{align}
except with exponentially small failure probability.

The second quantity is the maximal
tripartite product amplitude of $G$,
\begin{align}
\max_\alpha\lambda_\alpha(G)
=
\max_{\alpha,\gamma,u}
|\langle\alpha,\gamma,u|G\rangle|.
\end{align}
Applying the one-tensor result
of Lemma~\ref{lem:overlap} to
the three legs $A,C_A,R$, we obtain
\begin{align}
\Pr_G\left[
\max_\alpha\lambda_\alpha(G)
>
K\left(
\frac{1}{\sqrt{a c_A}}
+
\frac{1}{\sqrt{a r}}
+
\frac{1}{\sqrt{c_A r}}
\right)
\right]
\leq
e^{-\eta(a+c_A+r)}.
\label{eq:two-tensor-G-tripartite-overlap}
\end{align}

Thus, except with exponentially small
failure probability,
\begin{align}
\max_\alpha\lambda_\alpha(G)
\leq
K\left(
\frac{1}{\sqrt{a c_A}}
+
\frac{1}{\sqrt{a r}}
+
\frac{1}{\sqrt{c_A r}}
\right).
\label{eq:two-tensor-G-tripartite-simplified}
\end{align}

Define
\begin{align}
m:=\min\{a,b,c_A,c_B,r\}.
\end{align}
Choosing $t=\sqrt m$ in
Eq.~\eqref{eq:two-tensor-alpha-beta-optimized},
and combining the conditional probability
bound over $H$ with the bounds over $G$,
we obtain
\begin{align}
&\sqrt r\,
\max_{\alpha,\beta}
\bigl\|\,|G_\alpha\star H_\beta\rangle\,\bigr\|
\nonumber\\
&\quad\leq
K\left[
\frac{1}{\sqrt b}
\left(
\frac{1}{\sqrt a}
+
\frac{1}{\sqrt{c_A r}}
\right)
+
\frac{\sqrt a+\sqrt b}{\sqrt{b c_B}}
\left(
\frac{1}{\sqrt{a c_A}}
+
\frac{1}{\sqrt{a r}}
+
\frac{1}{\sqrt{c_A r}}
\right)
\right]
\label{eq:two-tensor-A-oriented-amplitude}
\end{align}
except with probability at most $e^{-\eta m}$,
after adjusting the universal constants.

Squaring Eq.~\eqref{eq:two-tensor-A-oriented-amplitude}
and absorbing numerical factors into $K$ gives
\begin{align}
\Lambda(\rho_{A:B})
\leq
K\left[
\frac{1}{ab}
+
\frac{1}{b c_A c_B}
+
\frac{1}{b c_B r}
+
\frac{a}{b c_A c_B r}
+
\frac{1}{a c_A c_B}
+
\frac{1}{c_A c_B r}
\right].
\label{eq:two-tensor-A-oriented-overlap}
\end{align}
Here, we have used the local triangle inequalities
$a<c_A r$ and $b<c_B r$ to absorb some
additional terms into the displayed ones.

\subsection{Optimization order and final bound}

So far, we have first averaged over $H$
while keeping $G$ fixed, and subsequently
optimized over $\beta$ and $\alpha$.

We can instead apply the same argument in
the opposite order.
In particular, we first average over $G$
while keeping $H$ fixed, optimize over $\alpha$,
and finally optimize over $\beta$.
This amounts to the exchange
\begin{align}
A\leftrightarrow B,
\qquad
C_A\leftrightarrow C_B,
\qquad
G\leftrightarrow H.
\end{align}
We therefore obtain
\begin{align}
\Lambda(\rho_{A:B})
\leq
K\left[
\frac{1}{ab}
+
\frac{1}{a c_A c_B}
+
\frac{1}{a c_A r}
+
\frac{b}{a c_A c_B r}
+
\frac{1}{b c_A c_B}
+
\frac{1}{c_A c_B r}
\right].
\label{eq:two-tensor-H-oriented-overlap-simplified}
\end{align}

Both bounds hold with failure probability
at most $e^{-\eta m}$, up to changes
in the universal constants.

We now choose the better of the two bounds.
Suppose first that
\begin{align}
b\geq a.
\end{align}
We use Eq.~\eqref{eq:two-tensor-A-oriented-overlap}.
Since $b<c_B r$ by the local triangle inequality,
we have
\begin{align}
\frac{1}{b c_B r}
\leq
\frac{1}{ab}.
\end{align}
Moreover, $a\leq b$ implies
\begin{align}
\frac{a}{b c_A c_B r}
\leq
\frac{1}{c_A c_B r}.
\end{align}
Consequently,
\begin{align}
\Lambda(\rho_{A:B})
\leq
K\left[
\frac{1}{ab}
+
\frac{1}{a c_A c_B}
+
\frac{1}{b c_A c_B}
+
\frac{1}{c_A c_B r}
\right].
\label{eq:two-tensor-final-overlap-bound}
\end{align}

If instead $a\geq b$, we use
Eq.~\eqref{eq:two-tensor-H-oriented-overlap-simplified}.
The same argument, with $A$ and $B$ exchanged,
leads to Eq.~\eqref{eq:two-tensor-final-overlap-bound}.

Thus, for either ordering of the external
dimensions, we obtain
\begin{align}
\Pr\left[
\Lambda(\rho_{A:B})
>
K\left(
\frac{1}{ab}
+
\frac{1}{a c_A c_B}
+
\frac{1}{b c_A c_B}
+
\frac{1}{c_A c_B r}
\right)
\right]
\leq
e^{-\eta m}.
\end{align}
This proves Lemma~\ref{lemma:2TN}.

The four terms correspond to the four
candidate tripartition areas
\begin{align}
A+B,\qquad
A+C_A+C_B,\qquad
B+C_A+C_B,\qquad
C_A+C_B+R,
\end{align}
respectively.
Together with the geometric lower bound
on the maximal product overlap from
Sec.~\ref{sec:bound}, this establishes
\begin{align}
-\log_2\Lambda(\rho_{A:B})
=
\mathrm{Area}(\Gamma_{\min})+O(1)
\end{align}
with high probability.

The comparison of the two optimization orders
is essential for obtaining the optimal bound.
A fixed optimization order can generate
additional terms, while choosing the better
order recovers precisely the candidate
minimal tripartition surfaces.

\section{Maximal product overlap for three Haar random tensors}
\label{app:three-tensor-overlap}

In this appendix, we prove Lemma~\ref{lemma:3TN}.
The proof follows the conditional Gaussian concentration
and sequential optimization method developed in
Appendix~\ref{app:two-tensor-overlap}.
The main additional step is to reduce the remaining
optimization problems to bipartite overlaps and
the two-tensor product-overlap bound.

\subsection{Conditional Gaussian concentration}

We begin by estimating
\begin{align}
\sqrt{r c_A c_B}\,
\bigl\|\,
|G_\alpha\star H_\beta\star J\rangle
\,\bigr\|
\end{align}
for fixed normalized states $\alpha$ and $\beta$.
We first average only over the Gaussian tensor $H$,
while keeping $G$ and $J$ fixed.

Recall the projected tensors
\begin{align}
|G_\alpha\rangle_{C_A R}
&:=
{}_A\langle\alpha|G\rangle_{A C_A R},
\\
|H_\beta\rangle_{C_B R}
&:=
{}_B\langle\beta|H\rangle_{B C_B R},
\end{align}
and define
\begin{align}
|G_\alpha\star J\rangle_{C C_B R}
:=
{}_A\langle\alpha|
G\star J\rangle_{A C C_B R}.
\end{align}

Averaging over $H$, we obtain
\begin{align}
\mathbb E_H\left[
r c_A c_B\,
\bigl\|\,
|G_\alpha\star H_\beta\star J\rangle
\,\bigr\|^2
\right]
=
\frac{c_A}{b}
\bigl\|\,
|G_\alpha\star J\rangle
\,\bigr\|^2.
\label{eq:three-tensor-fixed-alpha-beta-second-moment}
\end{align}
By Jensen's inequality,
\begin{align}
&\mathbb E_H\left[
\sqrt{r c_A c_B}\,
\bigl\|\,
|G_\alpha\star H_\beta\star J\rangle
\,\bigr\|
\right]
\nonumber\\
&\qquad\leq
\sqrt{\frac{c_A}{b}}\,
\bigl\|\,
|G_\alpha\star J\rangle
\,\bigr\|.
\label{eq:three-tensor-fixed-alpha-beta-mean}
\end{align}

The remaining norm is a bipartite product amplitude
of the effective two-tensor network $G\star J$,
\begin{align}
\bigl\|\,
|G_\alpha\star J\rangle
\,\bigr\|
=
\max_\phi
|\langle\alpha,\phi|G\star J\rangle|
=
\figbox{2.5}{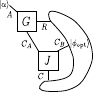},
\label{eq:three-tensor-expectation-overlap}
\end{align}
where $\phi$ is an arbitrary state on the combined
Hilbert space $C C_B R$.

To apply Gaussian concentration, define
\begin{align}
|h\rangle_{C_B R}
:=
\sqrt{b c_B r}\,
|H_\beta\rangle_{C_B R}.
\end{align}
By the covariance of $H$, this is a
standard complex Gaussian vector.
We introduce
\begin{align}
f_{\alpha,\beta;G,J}(h)
:=
\sqrt{\frac{c_A}{b}}\,
\bigl\|\,
|G_\alpha\star J\star h\rangle
\,\bigr\|.
\label{eq:three-tensor-fixed-product-function}
\end{align}
By construction,
\begin{align}
f_{\alpha,\beta;G,J}(h)
=
\sqrt{r c_A c_B}\,
\bigl\|\,
|G_\alpha\star H_\beta\star J\rangle
\,\bigr\|.
\end{align}

To estimate the Lipschitz constant, define
\begin{align}
\lambda_\alpha(G,J)
:=
\max_{\gamma,u}
|\langle\gamma,u|G_\alpha\star J\rangle|
=
\figbox{2.5}{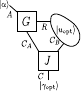},
\label{eq:three-tensor-lambda-definition}
\end{align}
where $\gamma\in\mathcal H_C$ and
$u\in\mathcal H_{C_B}\otimes\mathcal H_R$.
Importantly, $u$ is an arbitrary state on
the combined subsystem $C_B R$ and
need not factorize across $C_B:R$.

For arbitrary vectors $h,h'$ on $C_B R$,
let $\delta h=h-h'$.
The reverse triangle inequality gives
\begin{align}
&\left|
f_{\alpha,\beta;G,J}(h)
-
f_{\alpha,\beta;G,J}(h')
\right|
\nonumber\\
&\qquad\leq
\sqrt{\frac{c_A}{b}}\,
\bigl\|\,
|G_\alpha\star J\star\delta h\rangle
\,\bigr\|
\nonumber\\
&\qquad\leq
\sqrt{\frac{c_A}{b}}\,
\lambda_\alpha(G,J)\,
\|\delta h\|.
\end{align}
The second inequality follows because
$\lambda_\alpha(G,J)$ is the operator norm
of the contraction map from $C_B R$ to $C$,
exactly as in the two-tensor argument of
Appendix~\ref{app:two-tensor-overlap}.

Thus, $f_{\alpha,\beta;G,J}$ is
$L_\alpha$-Lipschitz with
\begin{align}
L_\alpha
=
\sqrt{\frac{c_A}{b}}\,
\lambda_\alpha(G,J).
\label{eq:three-tensor-Lipschitz-constant}
\end{align}

Combining this with the conditional expectation
Eq.~\eqref{eq:three-tensor-fixed-alpha-beta-mean},
Gaussian concentration gives
\begin{align}
&\Pr_H\left[
\sqrt{r c_A c_B}\,
\bigl\|\,
|G_\alpha\star H_\beta\star J\rangle
\,\bigr\|
\right.
\nonumber\\
&\hspace{16mm}\left.
>
\sqrt{\frac{c_A}{b}}
\left(
\bigl\|\,
|G_\alpha\star J\rangle
\,\bigr\|
+
t\,\lambda_\alpha(G,J)
\right)
\right]
\leq
e^{-\eta t^2}.
\label{eq:three-tensor-fixed-alpha-beta-tail}
\end{align}

The key point is that the bound holds for
every fixed $G$ and $J$.
The expectation and the fluctuations are
controlled by two overlap quantities involving
only the remaining $G$--$J$ subnetwork.

\subsection{Sequential optimization}

We next optimize over $\beta$ and $\alpha$.
As in Appendix~\ref{app:two-tensor-overlap},
we use $\epsilon$-nets of normalized states
on the Hilbert spaces of dimensions $b$ and $a$,
respectively.

For fixed $G$ and $J$, these nets can be
chosen with cardinalities
\begin{align}
|\mathcal N_B|\leq e^{Kb},
\qquad
|\mathcal N_A|\leq e^{Ka}.
\end{align}
Applying Eq.~\eqref{eq:three-tensor-fixed-alpha-beta-tail}
to the representative states and using
the union bound introduces deviations of order
$\sqrt b$ and $\sqrt a$, respectively.

The same net-to-sphere argument as in
Appendix~\ref{app:two-tensor-overlap}
then gives
\begin{align}
&\Pr_H\left[
\sqrt{r c_A c_B}\,
\max_{\alpha,\beta}
\bigl\|\,
|G_\alpha\star H_\beta\star J\rangle
\,\bigr\|
\right.
\nonumber\\
&\hspace{12mm}\left.
>
K\sqrt{\frac{c_A}{b}}
\left(
\max_\alpha
\bigl\|\,
|G_\alpha\star J\rangle
\,\bigr\|
+
(\sqrt a+\sqrt b+t)
\max_\alpha\lambda_\alpha(G,J)
\right)
\right]
\leq
e^{-\eta t^2}.
\label{eq:three-tensor-alpha-beta-optimized}
\end{align}

Importantly, this is a conditional probability
bound over $H$ for every fixed $G$ and $J$.
We can therefore first optimize over
$\alpha$ and $\beta$ using the randomness of $H$,
and then apply probability estimates to
the remaining $G$--$J$ subnetwork.

The problem has now been reduced to two quantities,
\begin{align}
T_1
&:=
\max_\alpha
\bigl\|\,
|G_\alpha\star J\rangle
\,\bigr\|,
\\
T_2
&:=
\max_\alpha
\lambda_\alpha(G,J).
\end{align}
We evaluate them separately below.

\subsection{Reduction to smaller tensor networks}

\paragraph{Bipartite overlap $T_1$.}

We first consider
\begin{align}
T_1
=
\max_\alpha
\bigl\|\,
|G_\alpha\star J\rangle
\,\bigr\|.
\end{align}

The contraction of the $C_A$ leg with $J$
can amplify a vector by at most the
maximal bipartite product amplitude of $J$
across $C_A:(C C_B)$.
Consequently,
\begin{align}
T_1
\leq
\left(
\max_\alpha
\bigl\|\,|G_\alpha\rangle\,\bigr\|
\right)
\left(
\max_{u,\phi}
|\langle u,\phi|J\rangle|
\right),
\label{eq:three-tensor-T1-factorization}
\end{align}
where $u\in\mathcal H_{C_A}$ and
$\phi\in\mathcal H_C\otimes\mathcal H_{C_B}$.

The two factors are bipartite overlap
amplitudes of individual Gaussian tensors,
\begin{align}
\max_\alpha\|G_\alpha\|
&=
\figbox{2.5}{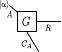},
\\
\max_{u,\phi}|\langle u,\phi|J\rangle|
&=
\figbox{2.5}{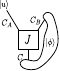}.
\end{align}

The standard bipartite Gaussian overlap bound
gives
\begin{align}
\max_\alpha
\bigl\|\,|G_\alpha\rangle\,\bigr\|
&\leq
K\left(
\frac{1}{\sqrt a}
+
\frac{1}{\sqrt{c_A r}}
\right),
\\
\max_{u,\phi}
|\langle u,\phi|J\rangle|
&\leq
K\left(
\frac{1}{\sqrt{c_A}}
+
\frac{1}{\sqrt{c c_B}}
\right)
\end{align}
with exponentially high probability.

Using the local triangle inequalities
$A<C_A+R$ and $C_A<C+C_B$,
these bounds simplify to
\begin{align}
\max_\alpha
\bigl\|\,|G_\alpha\rangle\,\bigr\|
\leq
\frac{K}{\sqrt a},
\qquad
\max_{u,\phi}
|\langle u,\phi|J\rangle|
\leq
\frac{K}{\sqrt{c_A}}.
\end{align}
Hence,
\begin{align}
T_1
\leq
\frac{K}{\sqrt{a c_A}}.
\label{eq:three-tensor-T1-bound}
\end{align}

\paragraph{Two-tensor overlap $T_2$.}

We next consider
\begin{align}
T_2
=
\max_\alpha
\lambda_\alpha(G,J)
=
\max_{\alpha,\gamma,u}
|\langle\gamma,u|G_\alpha\star J\rangle|.
\end{align}
Here, $\gamma$ is a state on $C$,
while $u$ is an arbitrary state on $C_B R$.

For a normalized state $\gamma\in\mathcal H_C$,
define
\begin{align}
|J_\gamma\rangle_{C_A C_B}
:=
{}_C\langle\gamma|J\rangle_{C C_A C_B}.
\end{align}
Optimizing over $u$ gives
\begin{align}
T_2
&=
\max_{\alpha,\gamma}
\bigl\|\,
|G_\alpha\star J_\gamma\rangle
\,\bigr\|
\nonumber\\
&=
\max_{\alpha,\gamma}
\left\|
\figbox{2.5}{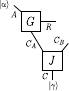}
\right\|.
\label{eq:three-tensor-T2-two-tensor}
\end{align}

This is precisely the two-tensor product-overlap
problem studied in
Sec.~\ref{sec:two-random-tensors}.
The two tensors are $G$ and $J$,
with external product legs $A,C$,
remaining legs $R,C_B$,
and shared internal bond $C_A$.

The normalization of this effective
two-tensor network introduces a factor $c_A$.
Applying Lemma~\ref{lemma:2TN}, we obtain
\begin{align}
c_A T_2^2
\leq
K\left(
\frac{1}{ac}
+
\frac{1}{a r c_B}
+
\frac{1}{c r c_B}
+
\frac{1}{r c_B c_A}
\right)
\label{eq:three-tensor-T2-general}
\end{align}
with exponentially high probability.

We now use the Mercedes condition
Eq.~\eqref{eq:three-tensor-Mercedes-phase}
together with the local triangle inequalities.
These imply
\begin{align}
A>\max\{C_A,R\},
\qquad
B>\max\{C_B,R\},
\qquad
C>\max\{C_A,C_B\}.
\label{eq:three-tensor-useful-inequalities}
\end{align}
For example, $A+B>R+C_A+C_B$
and $B<C_B+R$ imply $A>C_A$,
while $A+C>R+C_A+C_B$
and $C<C_A+C_B$ imply $A>R$.

In particular,
\begin{align}
a>c_A,
\qquad
c>c_A.
\end{align}
It follows that
\begin{align}
\frac{1}{a r c_B}
\leq
\frac{1}{r c_B c_A},
\qquad
\frac{1}{c r c_B}
\leq
\frac{1}{r c_B c_A}.
\end{align}
Thus, the second and third terms in
Eq.~\eqref{eq:three-tensor-T2-general}
are bounded by the final term.
We obtain
\begin{align}
T_2
\leq
K\left(
\frac{1}{\sqrt{a c c_A}}
+
\frac{1}{c_A\sqrt{r c_B}}
\right).
\label{eq:three-tensor-T2-bound}
\end{align}

\paragraph{Final bound in the Mercedes regime.}

We now substitute the estimates for
$T_1$ and $T_2$ into
Eq.~\eqref{eq:three-tensor-alpha-beta-optimized}.

By exchanging $A$ and $B$ if necessary,
we may assume
\begin{align}
a\leq b.
\end{align}
Recall that
\begin{align}
m=\min\{a,b,c,c_A,c_B,r\}.
\end{align}
Choosing $t=\sqrt m$, we have
\begin{align}
\sqrt a+\sqrt b+t
\leq
3\sqrt b.
\end{align}
Therefore, Eqs.~\eqref{eq:three-tensor-T1-bound}
and \eqref{eq:three-tensor-T2-bound} give
\begin{align}
&\sqrt{r c_A c_B}
\max_{\alpha,\beta}
\bigl\|\,
|G_\alpha\star H_\beta\star J\rangle
\,\bigr\|
\nonumber\\
&\qquad\leq
K\sqrt{\frac{c_A}{b}}
\left[
\frac{1}{\sqrt{a c_A}}
+
\sqrt b
\left(
\frac{1}{\sqrt{a c c_A}}
+
\frac{1}{c_A\sqrt{r c_B}}
\right)
\right]
\nonumber\\
&\qquad\leq
K\left(
\frac{1}{\sqrt{ab}}
+
\frac{1}{\sqrt{ac}}
+
\frac{1}{\sqrt{r c_A c_B}}
\right)
\end{align}
with exponentially high probability.

If instead $b\leq a$, we apply the same
argument with $A$ and $B$ exchanged.
Combining the two cases gives the symmetric bound
\begin{align}
\Lambda(\rho_{A:B})^{1/2}
\leq
K\left(
\frac{1}{\sqrt{ab}}
+
\frac{1}{\sqrt{ac}}
+
\frac{1}{\sqrt{bc}}
+
\frac{1}{\sqrt{r c_A c_B}}
\right)
\label{eq:three-tensor-four-amplitude-bound}
\end{align}
except with probability at most $e^{-\eta m}$,
after adjusting the universal constants.

Squaring this expression and absorbing
numerical factors into $K$, we obtain
\begin{align}
\Lambda(\rho_{A:B})
\leq
K\left(
\frac{1}{ab}
+
\frac{1}{ac}
+
\frac{1}{bc}
+
\frac{1}{r c_A c_B}
\right).
\label{eq:three-tensor-four-overlap-bound}
\end{align}

Finally, in the Mercedes phase,
\begin{align}
r c_A c_B<ab,
\qquad
r c_A c_B<ac,
\qquad
r c_A c_B<bc.
\end{align}
Therefore, every term in
Eq.~\eqref{eq:three-tensor-four-overlap-bound}
is bounded by $1/(r c_A c_B)$.
We conclude that
\begin{align}
\Lambda(\rho_{A:B})
\leq
\frac{K}{r c_A c_B}
\end{align}
except with probability at most $e^{-\eta m}$.
This proves Lemma~\ref{lemma:3TN}.

\mciteSetMidEndSepPunct{}{\ifmciteBstWouldAddEndPunct.\else\fi}{\relax}
\bibliographystyle{JHEP}
\bibliography{ref.bib}
\end{document}